\documentclass[conference]{IEEEtran}
\usepackage[switch]{lineno}
\usepackage{cite}
\usepackage{amsmath,amssymb,amsfonts}
\usepackage{algorithmic}
\usepackage{graphicx}
\usepackage{textcomp}
\usepackage{xcolor}
\def\BibTeX{{\rm B\kern-.05em{\sc i\kern-.025em b}\kern-.08em
    T\kern-.1667em\lower.7ex\hbox{E}\kern-.125emX}}

\usepackage[ruled, linesnumbered]{algorithm2e}
\usepackage{amsthm}
\usepackage{tabularx}
\usepackage{tikz}
\usepackage{subfigure}
\usetikzlibrary{decorations.pathreplacing}
\usetikzlibrary{patterns}

\newtheorem{theorem}{Theorem}
\newtheorem{lemma}[theorem]{Lemma}
\newtheorem{definition}{Definition}
\newtheorem{note}{Note}

\newcommand{\travel}[2]{\tau^{(#1)}_{#2}}
\newcommand{\travelmin}[2]{\underline{\tau}^{(#1)}_{#2}}
\newcommand{\travelmax}[2]{\overline{\tau}^{(#1)}_{#2}}
\newcommand{\dwell}[2]{\delta^{(#1)}_{#2}}
\newcommand{\dwellmin}[2]{\underline{\delta}^{(#1)}_{#2}}
\newcommand{\dwellmax}[2]{\overline{\delta}^{(#1)}_{#2}}
\newcommand{\dwellprime}[2]{\delta'^{(#1)}_{#2}}
\newcommand{\busstop}[1]{s^{(#1)}}
\newcommand{\holding}[2]{\phi^{(#1)}_{#2}}
\newcommand{\holdingprime}[2]{\phi'^{(#1)}_{#2}}
\newcommand{\holdingmin}[2]{\underline{\phi}^{(#1)}_{#2}}
\newcommand{\holdingmax}[2]{\overline{\phi}^{(#1)}_{#2}}
\newcommand{\arrival}[2]{a^{(#1)}_{#2}}
\newcommand{\arrivalprime}[2]{a'^{(#1)}_{#2}}
\newcommand{\arrivalmin}[2]{\underline{a}^{(#1)}_{#2}}

\newcommand{\arrivalmax}[2]{\overline{a}^{(#1)}_{#2}}
\newcommand{\arrivalmaxprime}[2]{\overline{a}'^{(#1)}_{#2}}
\newcommand{\estimatedarrival}[2]{\hat{a}^{(#1)}_{#2}}
\newcommand{\departure}[2]{d^{(#1)}_{#2}}
\newcommand{\departureprime}[2]{d'^{(#1)}_{#2}}
\newcommand{\departuremin}[2]{\underline{d}^{(#1)}_{#2}}
\newcommand{\departureminprime}[2]{\underline{d}'^{(#1)}_{#2}}
\newcommand{\departuremax}[2]{\overline{d}^{(#1)}_{#2}}

\newcommand{\release}[2]{r^{(#1)}_{#2}}
\newcommand{\carryin}[1]{\omega^{(#1)}}
\newcommand{\wcht}[1]{\overline{HT}^{(#1)}}
\newcommand{\bcht}[1]{\underline{HT}^{(#1)}}

\newcommand{\policy}[1]{\pi^{(#1)}}
\newcommand{\nullpolicy}{\mathtt{\pi NIL}}
\newcommand{\schedpolicy}{\mathtt{\pi SCH}}
\newcommand{\dynamicpolicy}{\mathtt{\pi DYN}}
\newcommand{\headway}[2]{h^{(#1)}_{#2}}

\begin{document}

\title{
Computing Headway Bounds under Worst-Case Bunching in Fixed-Line Transit Systems
}

\author{\IEEEauthorblockN{
Michael Yuhas$^{1,2}$, George Gunter$^{1}$, Jose Paolo Talusan$^{1}$, Aron Laszka$^{3}$, Dan Freudberg$^{4}$, Abhishek Dubey$^{1}$}
$^{1}$\textit{Institute for Software Integrated Systems}, \textit{Vanderbilt University}, Nashville, TN, USA\\
$^{2}$\textit{School of Science, Engineering and Technology}, \textit{St. Mary's University}, San Antonio, TX, USA\\
$^{3}$\textit{College of Information Sciences and Technology}, \textit{Pennsylvania State University}, University Park, PA, USA\\
$^{4}$\textit{WeGo Public Transit}, Nashville, TN, USA\\
}

\maketitle

\begin{abstract}
Vehicle bunching is a major problem for transit operators. 
When vehicles bunch together, the lead vehicle will service the majority of passenger demand, leaving the following vehicles to operate below capacity, wasting fuel and money. 
Furthermore, after the last vehicle in the bunch passes, the time before the next vehicle's arrival (headway) will be large. 
Transit operators can combat bunching by holding buses at stops along a route, trading riding time for even headway times. 
While prior work has focused on developing holding policies to minimize average case bunching, no work has focused on analyzing the longest and shortest possible headway times under a broad group of such policies. 
We assume that dwell times at stops and travel times between stops are bounded and develop a dynamic program that computes the maximum and minimum headway times for a single bus route with an arbitrary number of control points, vehicles, and holding policies. 
These bounds are tight in the sense that it is always possible to identify the specific sequence of events that lead to their occurrence. 
We use these bounds to investigate the effects of different holding policies, stop placement, and number of vehicles on route headways and worst-case bunching. 
Finally, we apply these analysis techniques to a real-world transit system in Nashville, TN and show their utility for transit planning. 
\end{abstract}


\section{Introduction}
\label{sec:intro}
In a fixed-line transit system, vehicles travel along predetermined routes and serve a set of predetermined stops, e.g., a bus line in an urban area.  
Fig.~\ref{fig:intro} provides a visualization of the stop order and a snapshot of the vehicle locations along such a route at a fixed point in time in an example system.  
While vehicles often depart from the terminal according to a fixed schedule, as they travel along their route, random variations in travel time and number of passengers lead to bunching: when many vehicles clump together and service stops in rapid succession.
A common pattern is when the lead vehicle encounters a large number of passengers boarding or alighting at a single stop: this delays the vehicle longer than expected, giving its follower a chance to catch up.  
Because less time has passed between two vehicles servicing the same stop, there are fewer passengers waiting to board at the subsequent stop and the follower spends less time waiting, causing the gap to progressively narrow~\cite{enayatollahi_modelling_2019}.  
This leads to vehicles clustering in \textit{bunches}, with the gap between bunches increasing beyond what is scheduled~\cite{pan_impact_2023}.  
This is bad from a transit operator’s perspective (the follower buses in a bunch are now operating below capacity) and from a rider’s perspective (if the rider arrives at a stop after the last vehicle in a bunch has passed, they must now wait longer to begin their journey).    

\begin{figure}
\centering
\resizebox{1\linewidth}{!}{
\begin{tikzpicture}
\draw[rounded corners=2pt] (-0.1,0.4) rectangle ++(0.2,0.7);
\draw[->] (0.1, 1) -- (1.9,1);
\draw (2, 1) circle (0.1);
\draw[->] (2.1, 1) -- (3.9,1);
\draw (4, 1) circle (0.1);
\draw[->] (4.1, 1) -- (5.9,1);
\draw (6, 1) circle (0.1);
\draw[->] (6.1, 1) -- (7.9,1);
\draw[rounded corners=2pt] (7.9,0.4) rectangle ++(0.2,0.7);

\node[anchor=south] at (8,1) {\includegraphics[width=20pt]{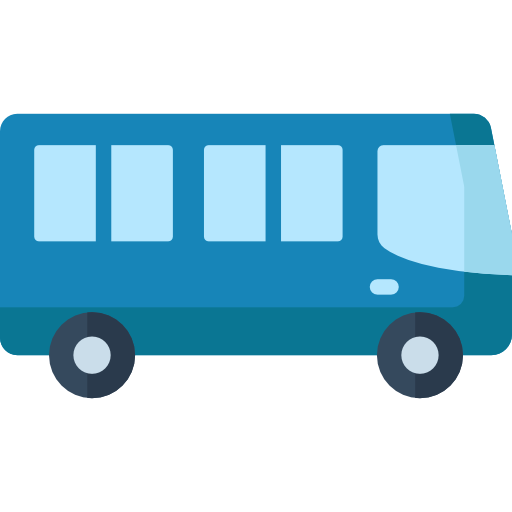}};
\node[anchor=south] at (7,1) {\includegraphics[width=20pt]{figures/bus.png}};
\node[anchor=south] at (6,1) {\includegraphics[width=20pt]{figures/bus.png}};
\node[anchor=south] at (0,1) {\includegraphics[width=20pt]{figures/bus.png}};
\node[anchor=north] at (4,0.5) {\scalebox{-1}[1]{\includegraphics[width=20pt]{figures/bus.png}}};
\node[anchor=north] at (5,0.5) {\scalebox{-1}[1]{\includegraphics[width=20pt]{figures/bus.png}}};

\draw[decorate,decoration={brace,amplitude=5pt,raise=2ex},color=red] (1,1) -- (5,1);
\node[anchor=south, align=center, color=red] (gap) at (3,1.4) {Increased Passenger \\ Waiting Times};
\draw[decorate,decoration={brace,amplitude=5pt,raise=4ex},color=red] (6,1.25) -- (8.5,1.25);
\node[anchor=south, color=red] (bunch) at (7.25, 2) {Bunched Vehicles};

\node[anchor=north] (terminal) at (0,0.5) {Terminal};
\node[anchor=north] (midpoint) at (8,0.5) {Midpoint};

\draw[<-] (5.1, 0.5) -- (7.9, 0.5);
\draw (5, 0.5) circle (0.1);
\draw[<-] (3.1, 0.5) -- (4.9, 0.5);
\draw (3, 0.5) circle (0.1);
\draw[<-] (0.1, 0.5) -- (2.9, 0.5);

\draw[rounded corners=4pt, color=green!40!gray] (4.5,-0.25) rectangle ++(1, 1);
\node[anchor=west, color=green!40!gray] (hold) at (5.5, -0.25) {Hold vehicle to prevent bunching};

\end{tikzpicture}
}
\caption{
As vehicles bunch, the headway for the vehicle behind the bunch increases.  Operators can hold vehicles at certain stops to regulate headway but cannot make vehicles speed up to close the gap.
}
\label{fig:intro}
\end{figure}

In fixed-line transit planning, the interval of time between vehicles at a stop is referred to as \textit{headway}~\cite{ntd_gloss}.  
In general, smaller headways are good for passengers, ensuring frequency of service, however, bunching leads to a scenario where some headway times are extremely small, leading to inefficiencies, and other headway times are extremely large, disgruntling riders.  
The goal then is, given a fixed number of vehicles serving a route, to even out headway times, such that a rider’s waiting time is bounded.
Transit operators can employ a variety of strategies to reduce bunching and manage headway.  
One strategy is to operate buses on a timetable at intermediate stops along the route, introducing delays for buses that arrive early~\cite{kim_development_2009}.  
With real-time vehicle location tracking, operators can also dynamically issue hold commands to vehicles getting too far ahead of their followers~\cite{pangilinan_bus_2008}.  
The hold times can be determined by a human operator, rules-based logic, or data driven model~\cite{rezazada_public_2022}. 
While holding a vehicle leads to an increase in in-vehicle travel time for passengers already aboard a bus, studies have shown that passengers perceive wait time at stops as worse than time spent traveling after boarding~\cite{he_travel_2019}.  
Furthermore, in areas with extreme weather (heat in the summer and cold in the winter), it can be dangerous for riders to wait for long periods at unsheltered stops.  
Thus, many transit agencies choose to pursue such policies to improve rider experience and improve operational efficiency.

Prior analysis of headway times has focused heavily on statistical models: travel times between stops and dwell times at each stop are drawn from distributions fit to real-world data with the objective of optimizing the average case~\cite{liu_improving_2018}.  
While this is important, human perception of a transit system is not necessarily shaped by the average case behavior~\cite{carrel_passengers_2013} and transit operators cannot make operational decisions based only on average case behavior.  
Thus, in addition to providing an analysis of the average case performance of various holding policies, it is necessary to provide a bound (both upper and lower) on headway times at each stop.  
While these times are potentially unbounded (e.g. due to vehicle breakdown), during normal operation it is safe to assume bounded dwell times at stops~\cite{dueker_determinants_2004} and travel times between stops~\cite{rahman_analysis_2018}.
Furthermore, if dwell or travel times from the long tail of these distributions occur, they are easy to detect and flag as violations.

\textbf{Contributions of this paper} -- 
We propose an analytical bound on the maximum and minimum headway times for a single transit route with an arbitrary number of stops, each governed by an arbitrary holding policy and serviced by an arbitrary number of vehicles.  
This work applies real-time system analysis to a cyber-physical system: automated headway control policies gather data through remote sensing and effect a change in a physical system by holding vehicles, which in turn affects the timing performance of system. 
To the best of our knowledge this is the first work of its kind: to place analytical bounds on headway times under broad classes of holding policies.
This work allows transit operators to evaluate their current holding policies and guide the selection of new proposals to ensure a safer and more enjoyable ride for all passengers as well as more economical operation.   
We use our bounds to provide an analysis of the tradeoffs involved in selecting a policy and include a case study using data from a real bus route in Nashville, TN. 

\section{Background}
The terminology and metrics used in this paper have been defined differently by different authors in slightly different transit contexts.  
\textit{Bunching} itself is seldom formally defined (definitions on how many and how far apart vehicles must be to constitute a bunch vary) but all authors agree that it can be measured indirectly through the unevenness of \textit{headway} times~\cite{byon_bunching_2018}.  
We use the Federal Transit Administration's definition of headway: ``the time interval between vehicles moving in the same direction on a particular route''~\cite{ntd_gloss}.  
Other authors use this working definition, but the expected headway at stops is the only location where headway directly impacts system performance~\cite{esfeh_waiting_2021}.  
In literature, \textit{tailway} is also used to describe the time between a lead vehicle and its follower~\cite{mohammed_influence_2021}, but this terminology is not as popular.

Many works focus on the problem of managing headway times.  
Infrastructure enhancements, such as traffic signal priority for buses can give vehicles falling behind schedule a chance to regulate headway~\cite{long_headway_2020}, however we do not consider them in this work as the required infrastructure is not readily available in every locality.  
Other works consider using a headway control tool to give commands to speed up or slow down to drivers at any point along route~\cite{martinez_understanding_2023}, but the commands may not always be carried out depending on driver skills and traffic conditions. 
The bulk of prior research on mitigating bunching focuses on holding buses at specific stops along a route called \textit{control points}~\cite{rezazada_public_2022}.  
We focus on analyzing these techniques since they are relatively easy for operators to implement without costly infrastructure upgrades and whether or not a driver holds is not subject to external factors like traffic.  
In general, these fall into two categories: schedule-based holding and headway-based holding.  
In schedule-based holding policies, operators focus on creating a schedule for buses that provides enough slack to mitigate headway variations, while keeping vehicle frequency high and travel times low~\cite{kim_development_2009}.  
The other approach is to use information about the headway between vehicles to dynamically adjust holding time based on real-world conditions.  
Daganzo provided a stochastic analysis of holding based only on headway between two vehicles on a fixed line transit route~\cite{daganzo_a_2009}.   
Cats et al. developed a simulator to track per stop waiting time distributions and considered two dynamic headway-based policies: always maintaining a minimum headway at a stop and holding so as to even headway times across the route~\cite{cats_evaluation_2010}.  
Zhou et al. considered holding based on the ratio of a vehicle’s headway to tailway and included a minimum and maximum holding time for operational convenience~\cite{zhou_a_2022}.  
Fonzone et al. extend these works by considering the case when travel and dwell times are dependent on non-uniform passenger arrival rates~\cite{fonzone_a_2015}.  
Recently, works have also considered using reinforcement learning to determine holding time based on headway~\cite{chen_real_2016}.  
In all these studies, however, the priority is always optimizing the average case behavior, not the headway bounds. 

\section{System Model}
\label{sec:model}
We now present a mathematical model of a fixed-line transit route.
While this work focuses on bus transit systems, the same model generalizes to other fixed-line transit modalities.
A table of notation is provided in Table~\ref{tab:notation}.
We use bold letters to represent a vector of values and italic text with subscripts to represent individual elements, e.g., $x_i$ is the $i$th value in $\mathbf{x}$.
For compactness, we use the notation $[\cdot]^{ub}_{lb}$ to represent a quantity clamped between ${ub}$ and ${lb}$, i.e., $[x]^{ub}_{lb}\triangleq\max\{lb,\min\{ub,x\}\}$.
We also indicate the upper bound of a variable with an overline and a lower bound with and underline (e.g. $\underline{x}\leq x\leq\overline{x}$).

\begin{table}
\caption{Table of notation.  For consistency, a superscript always denotes stop index and subscript always denotes vehicle index.}
\label{tab:notation}
\begin{tabularx}{\linewidth}{|p{1.6cm}|X|}
\hline
\textbf{Symbol} & \textbf{Definition} \\
\hline
$N$ & Number of stops in a route\\
$\busstop{i}$ & The $i$th stop along a route\\
$M$ & Number of vehicles serving a route\\
$\release{0}{j}$ & The time the $j$th vehicle begins service at stop $\busstop{0}$\\ 
$\travelmin{i}{}/\travelmax{i}{}$ & The minimum/maximum travel time for the route segment between stops $\busstop{i}$ and $\busstop{i+1}$\\
$\travel{i}{j}$ & The travel time for the $j$th vehicle to depart from $\busstop{i}$  on the route segment between stops $\busstop{i}$ and $\busstop{i+1}$\\
$\dwellmin{i}{}/\dwellmax{i}{}$ & The minimum/maximum dwell time at stop $\busstop{i}$\\
$\dwell{i}{j}$ & The dwell time for the $j$th vehicle to arrive at $\busstop{i}$\\
$\policy{i}(\cdot)$ & The policy in effect at stop $\busstop{i}$\\
$\holdingmin{i}{j}/\holdingmax{i}{j}$ & The minimum/maximum policy holding time for all buses at stop $\busstop{i}$\\
$\holding{i}{j}$ & The delay imposed by holding policy $\policy{i}$ for the $j$th vehicle to arrive at stop $\busstop{i}$\\
$\arrivalmin{i}{j}/\arrival{i}{j}/\arrivalmax{i}{j}$ & The earliest/actual/latest arrival time of the $j$th vehicle at $\busstop{i}$\\
$\estimatedarrival{i}{j}$ & The estimated time of the $j$th arrival at $\busstop{i}$\\
$\departuremin{i}{j}/\departure{i}{j}/\departuremax{i}{j}$ & The earliest/actual/latest departure time of the $j$th vehicle from $\busstop{i}$\\
$\headway{i}{j}$ & The $j$th vehicle's headway at $\busstop{i}$\\
$\wcht{i}$ & The upper bound on headway time at stop $\busstop{i}$\\
$\bcht{i}$ & The lower bound on headway at $\busstop{i}$\\
\hline
\end{tabularx}
\end{table}

\subsection{Transit System}
We consider a route as a sequence of $N$ stops $\busstop{0}\dots \busstop{N}$.
Each stop has an associated dwell time, which is the time a vehicle spends at the stop to load and unload passengers.
This is typically assumed to be drawn from a distribution~\cite{rashidi_estimating_2023}, which we assume has bounded support $[\dwellmin{i}{},\dwellmax{i}{}]$ at stop $\busstop{i}$ where $\dwellmin{i}{}$ is the minimum possible dwell time and $\dwellmax{i}{}$ is the maximum possible dwell time.
In the real world, dwell time is always lower bounded by at least $\dwellmin{i}{}=0$ (corresponding to the case where a driver skips a stop) and an unbounded dwell time ($\dwellmax{i}{}=\infty$) implies a bus break down or some other extreme event that prevents a bus from completing its route.
In practice, we care about the upper bound on dwell time that corresponds to some percentile of the underlying distribution.
This allows us to consider how well a holding policy performs given that dwell times conform to a higher density region of the underlying distribution and assign a probability that our analytical bound will not be exceeded during operation.
Likewise, we assume the travel time between stops $\busstop{i}$ and $\busstop{i+1}$ are drawn from the interval $[\travelmin{i}{},\travelmax{i}{}]$.
By the laws of physics (and most countries) a lower bound must exist and $\travelmin{i}{}>0$.
In the real world $\travelmax{i}{}$ could be unbounded due to an extreme event like a breakdown, but given a past distribution of travel times between stops, it is possible to select $\travelmax{i}{}$ to cover any arbitrary percentile of travel times.
Prior studies have considered bounded travel and dwell times~\cite{dueker_determinants_2004, rahman_analysis_2018}.

We consider a fleet of $M$ vehicles serving stops along the route.
Each vehicle must serve stops $\busstop{0}\dots \busstop{N}$ in order.
We consider the case where vehicles can overtake each other either during travel or at stops, which is possible and common in bus transit systems~\cite{saw_bus_2019}.
We also assume no vehicles are added or removed while serving the route.
While operators generally add vehicles for peak hours, the following analysis can simply be redone with the new number of vehicles.  We assume each vehicle is released into service at stop $\busstop{0}$ at time $\release{0}{j}$ where $j$ is the order in which the vehicle enters service (i.e., $j<j'\implies\release{0}{j}<\release{0}{j'}$).
The arrival time of the $j$th vehicle at stop $\busstop{i}$ is given by $\arrival{i}{j}$.
Note that $j$ is the index of arrival since the start of service, so $j$ counts up from $0\dots\infty$.
Each stop is governed by a holding policy, which instructs the driver of the $j$th vehicle to wait at the stop for at least $\holding{i}{j}$.
The holding policy at the $i$th stop is a function $\policy{i}(\cdot):X\mapsto\mathbb{R}_+$ where $X$ is the system state.
We assume that $\holding{i}{j}$ is computed upon a vehicle's arrival $\arrival{i}{j}$ and cannot be subsequently updated.
This assumption is based on real-world constraints where a bus driver can only safely check for a policy induced hold when stopped and a hold request cannot be rescinded once issued.
Policy induced holding occurs concurrently with dwell time, thus the total time the $j$th arrival spends stationary at $\busstop{i}$ is $\max\{\dwell{i}{j},\holding{i}{j}\}$ where $\dwell{i}{j}\in[\dwellmin{i}{},\dwellmax{i}{}]$.
We let $\departure{i}{j}$ represent the $j$th vehicle departure at $\busstop{i}$.
Letting $\travel{i}{j}\in[\travelmin{i}{},\travelmax{i}{}]$ be the travel time of the $j$th departure from $\busstop{i}$ to $\busstop{i+1}$, the progression of a single vehicle from its arrival at $\busstop{i}$ to its arrival at the subsequent stop is then given by Eqs.~(\ref{eq:departure}) and~(\ref{eq:arrival}).
Note that the $j$th arrival is not necessarily the $j$th departure, so Eq.~(\ref{eq:departure}) is written in terms of the $j'$th departure -- likewise the $j''$th arrival in Eq.~(\ref{eq:arrival}) -- where $j,j',j''\in0\dots\infty$.
\begin{align}
\departure{i}{j'}=&\;\arrival{i}{j} + \max\{\dwell{i}{j},\holding{i}{j}\}
\label{eq:departure}\\
\arrival{i+1}{j''} =&\;\departure{i}{j'}+ \travel{i}{j}
\label{eq:arrival}
\end{align}
Fig.~\ref{fig:model} provides an illustration of the time for three consecutive vehicles to travel between $\busstop{i-1}$ and $\busstop{i}$.

\begin{figure}
\centering
\resizebox{1\linewidth}{!}{
\begin{tikzpicture}
\draw[->] (0,0) -- (10,0);
\draw[->] (0,1) -- (10,1);
\node[anchor=west] (x-axis) at (10,0) {$t$};
\node[anchor=east] (stop0) at (0,0) {$\busstop{i-1}$};
\node[anchor=east] (stop1) at (0,1) {$\busstop{i}$};

\draw[orange] (0,0) -- (1,1);
\draw[orange] (1,1) -- (2,1);
\draw[orange] (2,1) -- (3,2);
\draw[draw=black, fill=orange, fill opacity=0.5, text opacity=1] (0,-1) rectangle ++(1,1) node[pos=.5] {$\travel{i-1}{j-1}$};
\draw[draw=black, fill=orange, fill opacity=0.5, text opacity=1] (1,-0.5) rectangle ++(1,0.5) node[pos=.5] {\footnotesize$\dwell{i}{j-1}$};
\draw[draw=black, fill=orange, fill opacity=0.5, text opacity=1] (1, -1) rectangle ++(0.75, 0.5) node[pos=.5] {\footnotesize$\holding{i}{j-1}$};
\draw[draw=black, fill=orange, fill opacity=0.5, text opacity=1] (2,-1) rectangle ++(1,1) node[pos=.5] {$\travel{i}{j-1}$};

\draw[blue] (1.75,0) -- (4.25,1);
\draw[blue] (4.25,1) -- (5.25,1);
\draw[blue] (5.25,1) -- (7.25,2);
\draw[draw=black, fill=blue, fill opacity=0.5, text opacity=1] (1.75,-2) rectangle ++(2.5,1) node[pos=.5] {$\travel{i-1}{j}$};
\draw[draw=black, fill=blue, fill opacity=0.5, text opacity=1] (4.25,-1.5) rectangle ++(0.75,0.5) node[pos=.5] {\footnotesize$\dwell{i}{j}$};
\draw[draw=black, fill=blue, fill opacity=0.5, text opacity=1] (4.25,-2) rectangle ++(1,0.5) node[pos=.5] {\footnotesize$\holding{i}{j}$};
\draw[draw=black, fill=blue, fill opacity=0.5, text opacity=1] (5.25,-2) rectangle ++(2,1) node[pos=.5] {$\travel{i}{j}$};

\draw[red] (6.0,0) -- (7.75,1);
\draw[red] (7.75,1) -- (8.75,1);
\draw[red] (8.75,1) -- (9.75,2);
\draw[draw=black, fill=red, fill opacity=0.5, text opacity=1] (6,-1) rectangle ++(1.75,1) node[pos=.5] {$\travel{i-1}{j+1}$};
\draw[draw=black, fill=red, fill opacity=0.5, text opacity=1] (7.75,-0.5) rectangle ++(1,0.5) node[pos=.5] {\footnotesize$\dwell{i}{j+1}$};
\draw[draw=black, fill=red, fill opacity=0.5, text opacity=1] (7.75,-1) rectangle ++(0.75,0.5) node[pos=.5] {\footnotesize$\holding{i}{j+1}$};
\draw[draw=black, fill=red, fill opacity=0.5, text opacity=1] (8.75,-1) rectangle ++(1,1) node[pos=.5] {$\travel{i}{j+1}$};

\draw[dotted] (1, 2.5) -- (1, -2.25);
\node[anchor=north] (aleader) at (1,-2.25) {$\arrival{i}{j-1}$};
\draw[dotted] (2,2.5) -- (2,-2.25);
\node[anchor=north] (dleader) at (2,-2.25) {$\departure{i}{j-1}$};
\draw[dotted] (4.25,2.5) -- (4.25,-2.25);
\node[anchor=north] (a) at (4.25,-2.25) {$\arrival{i}{j}$};
\draw[dotted] (5.25,2.5) -- (5.25,-2.25);
\node[anchor=north] (d) at (5.25,-2.25) {$\departure{i}{j}$};
\draw[dotted] (7.75,2.5) -- (7.75,-2.25);
\node[anchor=north] (afollower) at (7.75,-2.25) {$\arrival{i}{j+1}$};
\draw[dotted] (8.75,2.5) -- (8.75,-2.25);
\node[anchor=north] (dfollower) at (8.75, -2.25) {$\departure{i}{j+1}$};

\draw[draw=black, pattern color=green, pattern=north east lines, text opacity=1] (2,2) rectangle ++(2.25,0.5) node[pos=0.5] {$j$'s Headway};
\draw[draw=black, pattern color=magenta, pattern=north west lines, text opacity=1] (5.25,2) rectangle ++(2.5,0.5) node[pos=0.5] {$j$'s Tailway};

\end{tikzpicture}
}
\caption{Illustration of headway and tailway calculation for the $j$th arrival at stop $\busstop{i}$.  Vehicle $j-1$ is the leader and $j+1$ is the follower making vehicle $j$'s headway $\headway{i}{j}=\arrival{i}{j}-\departure{i}{j-1}$ and the tailway $\headway{i}{j+1}=\arrival{i}{j+1}-\departure{i}{j}$.}
\label{fig:model}
\end{figure}
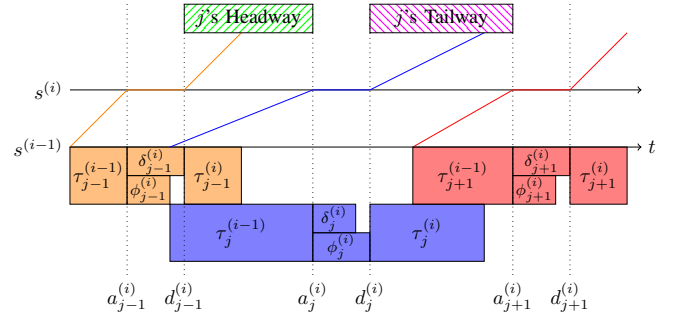

\subsection{Metrics}
Headway time refers to the spacing between vehicles, i.e., the duration between the departure of one vehicle and the arrival of the subsequent vehicle at the same station~\cite{ntd_gloss}.
For the purposes of this paper, headway time is only defined at stops, not every point along a route; riders are only impacted by headway as measured at a station and the economic impact of uneven headway times is only realized based on vehicle occupancy, which only changes at the stops.
\begin{definition}
The headway $\headway{i}{j}$ of the $j$th arrival at $\busstop{i}$ is given by Eq.~(\ref{eq:headway}).
\begin{equation}
\label{eq:headway}
\headway{i}{j}=\max\{0,\arrival{i}{j}-\departure{i}{j-1}\}
\end{equation}
\end{definition}
\begin{note}
$\headway{i}{}$ is undefined at each $\busstop{i}$ until the second vehicle arrives ($t\geq\arrival{i}{1}$).
Furthermore, at $\busstop{0}$ headway remains undefined until the release of the last vehicle into service ($t\geq\release{0}{M}$).
\end{note}
Fig.~\ref{fig:model} illustrates the concept of headway calculation.
For convenience, we also introduce the term \textit{tailway}, which refers to the headway of the subsequent arrival, i.e., $\headway{i}{j+1}$.
While many studies focus on computing the average headway time, we focus on computing headway time bounds.
\begin{definition}
$\wcht{i}$ is the upper bound on headway time that can occur at $\busstop{i}$ during the service period and $\bcht{i}$ is the lower bound, i.e.:
\begin{align}
\label{eq:wcht}
\wcht{i}=&\max_{0<j<\infty}\{\headway{i}{j}\}\\
\label{eq:bcht}
\bcht{i}=&\min_{0<j<\infty}\{\headway{i}{j}\}
\end{align}
\end{definition}
Given a fixed number of vehicles serving a route, bunching causes $\wcht{i}$ to increase and $\bcht{i}$ to decrease.
In areas where waiting for the bus means standing unsheltered in extreme cold or heat, a bound on maximum time spent waiting ($\wcht{i}$) may be more important than the average value.
Likewise, bounding minimum headway can ensure good bus utilization and efficient use of resources.

\subsection{Bus Holding Policies}
\label{sec:policies}
A holding policy guarantees that a vehicle remains stationary at a stop for a minimum amount of time regardless of $\dwell{i}{j}$.
This interval is determined by the holding policy at the stop $\policy{i}(\cdot)$, which computes $\holding{i}{j}$ for every vehicle at the moment it reaches the stop.
This reflects real world conditions where issuing policy updates when a driver is approaching the end of their hold period can result in the updates being ignored due to the notification not arriving on time or the driver being preoccupied with other tasks.
On the other hand, most control systems seek to apply actions as late as possible so as to minimize the affect of random variations in system state that happen between when a hold is requested and when it is executed.
Since the vehicle's dwell time at a stop is unknown and drawn from a distribution, arrival is the latest possible moment when this should be calculated.
Policies are applied on a per stop basis, not to individual vehicles and not across an entire route.
The reason for this is practical: it is not feasible for operators to idle vehicles at every stop, or the time may be limited due to infrastructure constraints.
Consider, for example, stops along a two-lane road with no turnouts: in these cases, a vehicle can still stop and dwell for passengers to alight and board, but idling for too long interferes with other traffic.
Furthermore, some stops along a route may require a vehicle to synchronize with other transit times that cannot be adjusted (e.g. commuter rail).
The space of possible policies is infinite, so we focus on three sets that encompass the policies most commonly deployed by transit operators.  In every case, a policy is a function that maps some set of state measurements to a positive real number (hold time), so machine learning-based policies fit into our analysis framework as long as they are a member of one of the following sets.

\subsubsection{Null Policy ($\nullpolicy$)}
If $\policy{i}\in\nullpolicy$ vehicles are never held, i.e., for stop $\busstop{i}$, $\holding{i}{j}=0\;\forall j\in0\dots\infty$.  This is the default policy for most stops in real world transit systems.
\subsubsection{Static Schedule-Based Policies ($\schedpolicy$)}
Under this set of policies vehicles depart stops at predetermined times.  We let $\mathbf{T}$ be a list of scheduled departure times.  Only one vehicle is permitted to depart from the stop at each of the entries in $\mathbf{T}$. If multiple vehicles are held at the stop waiting for the next scheduled departure, we assume they depart in the order of their arrival, i.e., a stop governed by this policy behaves as a FIFO queue.  If no vehicle departs at a time listed in $\mathbf{T}$, the next vehicle to arrive at the stop receives no holding time, corresponding to the case where vehicles are running behind schedule.  Mathematically, we model this as the stateful function shown in Eq.~(\ref{eq:pisched}) where $l$ is the index of the last used departure time in $\mathbf{T}$ and $t$ is the current time.  Note that if dwell time $\dwell{i}{j}$ exceeds the holding time $\holding{i}{j}$, it is still possible for a vehicle to depart after its scheduled release.
\begin{equation}
\label{eq:pisched}
\policy{i}(t;l,\mathbf{T}) = \max{\{0, \mathbf{T}_l-t\}}
\end{equation}
\subsubsection{Headway-Based Control ($\dynamicpolicy$)}
Under this set of policies the time a vehicle idles at a given stop is a function of the locations of other vehicles along the route.
The goal is to hold a vehicle so that the gap between headway is roughly equal across all arrival and departure pairs.
We assume that the holding calculated by headway-based control is only a function of the current (and past) positions of vehicles along a route, i.e., $\policy{i}(X):\mathbb{R}_+^{2N\times M}\mapsto\mathbb{R}_+$.
where $X$ refers to the system state -- vectors of the last arrival times $\mathbf{a}^{(i')}$ and departure times $\mathbf{d}^{(i')}$ for each of $M$ vehicles at stop $i'$, for $i'\in[0\dots N]$.
Note that it is not required for a dynamic policy to use all of this information and only a subset may be required.
For our analysis to work, we make the following assumptions about $\policy{i}\in\dynamicpolicy$:
\begin{enumerate}
    \item The output is bounded: $0\leq\holding{i}{j}<\infty$.
    \item The policy is 1-left-Lipschitz~\cite{duetting_optimal_2023} with respect to headway, i.e., when headway is computable from the system state, the following inequality must hold: $\policy{i}(\headway{i}{j})-\policy{i}(\headway{i}{j'})\leq \headway{i}{j'} - \headway{i}{j}\forall \headway{i}{j}\leq\headway{i}{j'}$.
\end{enumerate}
Assumption (1) is a requirement for any practical policy as unbounded $\holding{i}{j}$ would mean a bus is stuck at a stop forever.
Assumption (2) requires that a vehicle arriving with more tailway is not held longer than if it had arrived with less tailway.
Prior work has shown it is possible to train Lipschitz bounded~\cite{fazlyab_efficient_2019} neural networks, so $\dynamicpolicy$ still represents a broad class of data driven policies.
Many headway-based holding policies~\cite{daganzo_a_2009,zhou_a_2022} already satisfy these properties.

\textbf{Problem Statement} -- The goal of this paper is to identify the upper and lower bounds on headway times for all stops in any bus route modeled by the system described above.
This problem is challenging because headway bounds are dependent not just on the maximum and minimum travel and dwell times between two stops, but also all the possible travel and dwell times between any two stops taken by every vehicle servicing the route, at every time since the beginning of service, as these affect policy holding for stops governed by $\schedpolicy$ and $\dynamicpolicy$, which in turn affect $\wcht{i}$ and $\bcht{i}$.

\section{Headway Time Bounds}
We first show that in a route with all stops governed by $\policy{i}\in\nullpolicy$ (and in some cases $\schedpolicy$) headway is upper bounded by the maximum travel time and lower bounded by~$0$.
Finally, we use a dynamic program to establish tight upper and lower headway bounds for any route.
In order to accomplish this, we must formally define a \textit{bunch}. 
\begin{definition}
\label{def:bunch}
A bunch is a sequence of $v$ vehicles with leader arriving at $\busstop{i}$ at $\arrival{i}{j}$ and the last bus in the bunch departing the same $\busstop{i}$ at $\departure{i}{j+v-1}$.  The time for a bunch to pass the stop is given as $\departure{i}{j+v-1}-\arrival{i}{j}$.
\end{definition}

\subsection{Insufficiency of Static Policies}
\begin{theorem}
\label{thm:wcht=wchtt}
In a route with all stops governed by $\nullpolicy$, if $\travelmax{i}{}>\travelmin{i}{}$ or $\dwellmax{i}{}>\dwellmin{i}{}$ for at least one stop $\busstop{i}$, then $\bcht{k}=0$ and $\wcht{k}=\sum_{i'=0\dots N}(\travelmax{i'}{}+\dwellmax{i'}{})-\dwellmax{k}{}\forall k\in0\dots N$ (which is equivalent to the maximum time required to travel the entire route, including dwell time).
\end{theorem}
\begin{proof}
The maximum and minimum headway time occurs when all $M$ vehicles are bunched together such that $\exists j|\arrival{k}{j}=\arrival{k}{j+m}\forall m\in0\dots M$, i.e., all buses arrive and depart $\busstop{k}$ simultaneously ($\bcht{k}=0$).
Given this precondition, it is always possible for all $M$ vehicles to take the maximum travel time and dwell times for all stops back to $\busstop{k}$, i.e., all vehicles arrive at and depart from the remaining stops simultaneously.
This is functionally equivalent to a single vehicle servicing the route, in which case maximum headway is the time between the vehicle departing $\busstop{k}$ and arriving after traversing the entire route: $\wcht{k}=\sum_{i'=0\dots N}(\travelmax{i'}{}+\dwellmax{i'}{})-\dwellmax{k}{}$ (recall that under $\nullpolicy$, $\holding{i}{j}=0\;\forall i,j$).
To show that $\arrival{k}{j}=\arrival{k}{j+m}\forall m\in0\dots M$ is always possible, assume either $\travelmax{i}{}>\travelmin{i}{}$ or $\dwellmax{i}{}>\dwellmin{i}{}$ for at least one stop $\busstop{i}$ along the route.
Consider two vehicles at $\busstop{i}$: a leader $j-1$ and a follower $j$.
Under the system model they can take the same transit and dwell times for all other stops excluding $\busstop{i}$.
Let $\headway{i}{j}=\arrival{i}{j}-\departure{i}{j-1}$ denote the initial headway time between leader and follower.
Suppose that for the next $\lfloor \headway{i}{j}/(\travelmax{i}{}-\travelmin{i}{})\rfloor$ arrivals at $\busstop{i}$, the leader takes $\travelmax{i}{}$ to complete $\busstop{i
}$'s travel and the follower takes $\travelmin{i}{}$.
On the very next iteration, if the leader still takes $\travelmax{i}{}$ to complete the leg, and the follower takes $\travelmin{i}{}+\headway{i}{j}\mod (\travelmax{i}{}-\travelmin{i}{})$ the headway between them will be zero.
In the worst case, both vehicles now take the same time to complete each route segment.
The same reasoning can be applied for $\dwellmax{i}{}>\dwellmin{i}{}$.
\end{proof}
\begin{note}
\label{note:schedpolicy}
It follows from Theorem~\ref{thm:wcht=wchtt} that for a route with stop $\busstop{k}$ controlled by $\policy{k}\in\schedpolicy$, as long as $\sum_{i=k'\dots k}{(\travelmax{i}{}+\dwellmax{i}{}-\travelmin{i}{}-\dwellmin{i}{})}>\mathbf{T}_l-\mathbf{T}_{l-1}\forall{l}=0\dots\infty$ where $k'$ is the stop index after the last schedule controlled stop and all other stops are controlled by $\nullpolicy$, then the maximum headway time is equivalent to that given in Theorem~\ref{thm:wcht=wchtt}.
\end{note}
Note~\ref{note:schedpolicy} implies that even if transit operators use schedule-driven control at the terminal, it will not be sufficient to guarantee that headway is less than the route completion time if it is possible for vehicles to arrive behind schedule.
To build intuition, Fig.~\ref{fig:thmwcht=wcht} illustrates how bunching can lead to the bounds established in Theorem~\ref{thm:wcht=wchtt} and Note~\ref{note:schedpolicy} in a route with $3$ stops and $3$ buses.

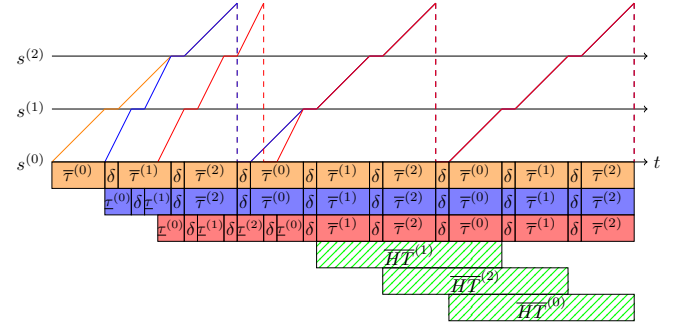
\begin{figure}
\centering
\resizebox{1\linewidth}{!}{
\begin{tikzpicture}
\draw[->] (0,0) -- (11.25,0);
\draw[->] (0,1) -- (11.25,1);
\draw[->] (0,2) -- (11.25,2);
\node[anchor=west] (x-axis) at (11.25,0) {$t$};
\node[anchor=east] (stop0) at (0,0) {$\busstop{0}$};
\node[anchor=east] (stop1) at (0,1) {$\busstop{1}$};
\node[anchor=east] (stop2) at (0,2) {$\busstop{2}$};

\draw[orange] (0,0) -- (1,1);
\draw[orange] (1,1) -- (1.25,1);
\draw[orange] (1.25,1) -- (2.25,2);
\draw[orange] (2.25,2) -- (2.5,2);
\draw[orange] (2.5,2) -- (3.5,3);
\draw[orange, dashed] (3.5,3) -- (3.5,0);
\draw[orange] (3.5,0) -- (3.75,0);
\draw[orange] (3.75,0) -- (4.75,1);
\draw[orange] (4.75,1) -- (5,1);
\draw[orange] (5,1) -- (6,2);
\draw[orange] (6,2) -- (6.25,2);
\draw[orange] (6.25,2) -- (7.25,3);
\draw[orange, dashed] (7.25,3) -- (7.25,0);
\draw[orange] (7.25,0) -- (7.5, 0);
\draw[orange] (7.5,0) -- (8.5, 1);
\draw[orange] (8.5,1) -- (8.75, 1);
\draw[orange] (8.75,1) -- (9.75, 2);
\draw[orange] (9.75,2) -- (10, 2);
\draw[orange] (10,2) -- (11,3);
\draw[orange, dashed] (11,3) -- (11,0);

\draw[blue] (1,0) -- (1.5,1);
\draw[blue] (1.5,1) -- (1.75,1);
\draw[blue] (1.75,1) -- (2.25,2);
\draw[blue] (2.25,2) -- (2.5,2);
\draw[blue] (2.5,2) -- (3.5,3);
\draw[blue, dashed] (3.5,3) -- (3.5,0);
\draw[blue] (3.5,0) -- (3.75,0);
\draw[blue] (3.75,0) -- (4.75,1);
\draw[blue] (4.75,1) -- (5,1);
\draw[blue] (5,1) -- (6,2);
\draw[blue] (6,2) -- (6.25,2);
\draw[blue] (6.25,2) -- (7.25,3);
\draw[blue, dashed] (7.25,3) -- (7.25,0);
\draw[blue] (7.25,0) -- (7.5, 0);
\draw[blue] (7.5,0) -- (8.5, 1);
\draw[blue] (8.5,1) -- (8.75, 1);
\draw[blue] (8.75,1) -- (9.75, 2);
\draw[blue] (9.75,2) -- (10, 2);
\draw[blue] (10,2) -- (11,3);
\draw[blue, dashed] (11,3) -- (11,0);

\draw[red] (2,0) -- (2.5,1);
\draw[red] (2.5,1) -- (2.75,1);
\draw[red] (2.75,1) -- (3.25,2);
\draw[red] (3.25,2) -- (3.5,2);
\draw[red] (3.5,2) -- (4,3);
\draw[red, dashed] (4,3) -- (4,0);
\draw[red] (4,0) -- (4.25,0);
\draw[red] (4.25,0) -- (4.75,1);
\draw[red] (4.75,1) -- (5,1);
\draw[red] (5,1) -- (6,2);
\draw[red] (6,2) -- (6.25,2);
\draw[red] (6.25,2) -- (7.25,3);
\draw[red, dashed] (7.25,3) -- (7.25,0);
\draw[red] (7.25,0) -- (7.5, 0);
\draw[red] (7.5,0) -- (8.5, 1);
\draw[red] (8.5,1) -- (8.75, 1);
\draw[red] (8.75,1) -- (9.75, 2);
\draw[red] (9.75,2) -- (10, 2);
\draw[red] (10,2) -- (11,3);
\draw[red, dashed] (11,3) -- (11,0);

\draw[draw=black, pattern color=green, pattern=north east lines, text opacity=1] (7.5,-3) rectangle ++(3.5,0.5) node[pos=0.5] {$\wcht{0}$};
\draw[draw=black, pattern color=green, pattern=north east lines, text opacity=1] (6.25,-2.5) rectangle ++(3.5,0.5) node[pos=0.5] {$\wcht{2}$};
\draw[draw=black, pattern color=green, pattern=north east lines, text opacity=1] (5,-2) rectangle ++(3.5,0.5) node[pos=0.5] {$\wcht{1}$};

\draw[draw=black, fill=orange, fill opacity=0.5, text opacity=1] (10,-0.5) rectangle ++(1,0.5) node[pos=0.5] {$\travelmax{2}{}$};
\draw[draw=black, fill=orange, fill opacity=0.5, text opacity=1] (9.75,-0.5) rectangle ++(0.25,0.5) node[pos=0.5] {$\delta$};
\draw[draw=black, fill=orange, fill opacity=0.5, text opacity=1] (8.75,-0.5) rectangle ++(1,0.5) node[pos=0.5] {$\travelmax{1}{}$};
\draw[draw=black, fill=orange, fill opacity=0.5, text opacity=1] (8.5,-0.5) rectangle ++(0.25,0.5) node[pos=0.5] {$\delta$};
\draw[draw=black, fill=orange, fill opacity=0.5, text opacity=1] (7.5,-0.5) rectangle ++(1,0.5) node[pos=0.5] {$\travelmax{0}{}$};
\draw[draw=black, fill=orange, fill opacity=0.5, text opacity=1] (7.25,-0.5) rectangle ++(0.25,0.5) node[pos=0.5] {$\delta$};
\draw[draw=black, fill=orange, fill opacity=0.5, text opacity=1] (6.25,-0.5) rectangle ++(1,0.5) node[pos=0.5] {$\travelmax{2}{}$};
\draw[draw=black, fill=orange, fill opacity=0.5, text opacity=1] (6,-0.5) rectangle ++(0.25,0.5) node[pos=0.5] {$\delta$};
\draw[draw=black, fill=orange, fill opacity=0.5, text opacity=1] (5,-0.5) rectangle ++(1,0.5) node[pos=0.5] {$\travelmax{1}{}$};
\draw[draw=black, fill=orange, fill opacity=0.5, text opacity=1] (4.75,-0.5) rectangle ++(0.25,0.5) node[pos=0.5] {$\delta$};
\draw[draw=black, fill=orange, fill opacity=0.5, text opacity=1] (3.75,-0.5) rectangle ++(1,0.5) node[pos=0.5] {$\travelmax{0}{}$};
\draw[draw=black, fill=orange, fill opacity=0.5, text opacity=1] (3.5,-0.5) rectangle ++(0.25,0.5) node[pos=0.5] {$\delta$};
\draw[draw=black, fill=orange, fill opacity=0.5, text opacity=1] (2.5,-0.5) rectangle ++(1,0.5) node[pos=0.5] {$\travelmax{2}{}$};
\draw[draw=black, fill=orange, fill opacity=0.5, text opacity=1] (2.25,-0.5) rectangle ++(0.25,0.5) node[pos=0.5] {$\delta$};
\draw[draw=black, fill=orange, fill opacity=0.5, text opacity=1] (1.25,-0.5) rectangle ++(1,0.5) node[pos=0.5] {$\travelmax{1}{}$};
\draw[draw=black, fill=orange, fill opacity=0.5, text opacity=1] (1,-0.5) rectangle ++(0.25,0.5) node[pos=0.5] {$\delta$};
\draw[draw=black, fill=orange, fill opacity=0.5, text opacity=1] (0,-0.5) rectangle ++(1,0.5) node[pos=0.5] {$\travelmax{0}{}$};

\draw[draw=black, fill=blue, fill opacity=0.5, text opacity=1] (10,-1) rectangle ++(1,0.5) node[pos=0.5] {$\travelmax{2}{}$};
\draw[draw=black, fill=blue, fill opacity=0.5, text opacity=1] (9.75,-1) rectangle ++(0.25,0.5) node[pos=0.5] {$\delta$};
\draw[draw=black, fill=blue, fill opacity=0.5, text opacity=1] (8.75,-1) rectangle ++(1,0.5) node[pos=0.5] {$\travelmax{1}{}$};
\draw[draw=black, fill=blue, fill opacity=0.5, text opacity=1] (8.5,-1) rectangle ++(0.25,0.5) node[pos=0.5] {$\delta$};
\draw[draw=black, fill=blue, fill opacity=0.5, text opacity=1] (7.5,-1) rectangle ++(1,0.5) node[pos=0.5] {$\travelmax{0}{}$};
\draw[draw=black, fill=blue, fill opacity=0.5, text opacity=1] (7.25,-1) rectangle ++(0.25,0.5) node[pos=0.5] {$\delta$};
\draw[draw=black, fill=blue, fill opacity=0.5, text opacity=1] (6.25,-1) rectangle ++(1,0.5) node[pos=0.5] {$\travelmax{2}{}$};
\draw[draw=black, fill=blue, fill opacity=0.5, text opacity=1] (6,-1) rectangle ++(0.25,0.5) node[pos=0.5] {$\delta$};
\draw[draw=black, fill=blue, fill opacity=0.5, text opacity=1] (5,-1) rectangle ++(1,0.5) node[pos=0.5] {$\travelmax{1}{}$};
\draw[draw=black, fill=blue, fill opacity=0.5, text opacity=1] (4.75,-1) rectangle ++(0.25,0.5) node[pos=0.5] {$\delta$};
\draw[draw=black, fill=blue, fill opacity=0.5, text opacity=1] (3.75,-1) rectangle ++(1,0.5) node[pos=0.5] {$\travelmax{0}{}$};
\draw[draw=black, fill=blue, fill opacity=0.5, text opacity=1] (3.5,-1) rectangle ++(0.25,0.5) node[pos=0.5] {$\delta$};
\draw[draw=black, fill=blue, fill opacity=0.5, text opacity=1] (2.5,-1) rectangle ++(1,0.5) node[pos=0.5] {$\travelmax{2}{}$};
\draw[draw=black, fill=blue, fill opacity=0.5, text opacity=1] (2.25,-1) rectangle ++(0.25,0.5) node[pos=0.5] {$\delta$};
\draw[draw=black, fill=blue, fill opacity=0.5, text opacity=1] (1.75,-1) rectangle ++(0.5,0.5) node[pos=0.5] {\footnotesize$\travelmin{1}{}$};
\draw[draw=black, fill=blue, fill opacity=0.5, text opacity=1] (1.5,-1) rectangle ++(0.25,0.5) node[pos=0.5] {$\delta$};
\draw[draw=black, fill=blue, fill opacity=0.5, text opacity=1] (1,-1) rectangle ++(0.5,0.5) node[pos=0.5] {\footnotesize$\travelmin{0}{}$};

\draw[draw=black, fill=red, fill opacity=0.5, text opacity=1] (10,-1.5) rectangle ++(1,0.5) node[pos=0.5] {$\travelmax{2}{}$};
\draw[draw=black, fill=red, fill opacity=0.5, text opacity=1] (9.75,-1.5) rectangle ++(0.25,0.5) node[pos=0.5] {$\delta$};
\draw[draw=black, fill=red, fill opacity=0.5, text opacity=1] (8.75,-1.5) rectangle ++(1,0.5) node[pos=0.5] {$\travelmax{1}{}$};
\draw[draw=black, fill=red, fill opacity=0.5, text opacity=1] (8.5,-1.5) rectangle ++(0.25,0.5) node[pos=0.5] {$\delta$};
\draw[draw=black, fill=red, fill opacity=0.5, text opacity=1] (7.5,-1.5) rectangle ++(1,0.5) node[pos=0.5] {$\travelmax{0}{}$};
\draw[draw=black, fill=red, fill opacity=0.5, text opacity=1] (7.25,-1.5) rectangle ++(0.25,0.5) node[pos=0.5] {$\delta$};
\draw[draw=black, fill=red, fill opacity=0.5, text opacity=1] (6.25,-1.5) rectangle ++(1,0.5) node[pos=0.5] {$\travelmax{2}{}$};
\draw[draw=black, fill=red, fill opacity=0.5, text opacity=1] (6,-1.5) rectangle ++(0.25,0.5) node[pos=0.5] {$\delta$};
\draw[draw=black, fill=red, fill opacity=0.5, text opacity=1] (5,-1.5) rectangle ++(1,0.5) node[pos=0.5] {$\travelmax{1}{}$};
\draw[draw=black, fill=red, fill opacity=0.5, text opacity=1] (4.75,-1.5) rectangle ++(0.25,0.5) node[pos=0.5] {$\delta$};
\draw[draw=black, fill=red, fill opacity=0.5, text opacity=1] (4.25,-1.5) rectangle ++(0.5,0.5) node[pos=0.5] {\footnotesize$\travelmin{0}{}$};
\draw[draw=black, fill=red, fill opacity=0.5, text opacity=1] (4,-1.5) rectangle ++(0.25,0.5) node[pos=0.5] {$\delta$};
\draw[draw=black, fill=red, fill opacity=0.5, text opacity=1] (3.5,-1.5) rectangle ++(0.5,0.5) node[pos=0.5] {\footnotesize$\travelmin{2}{}$};
\draw[draw=black, fill=red, fill opacity=0.5, text opacity=1] (3.25,-1.5) rectangle ++(0.25,0.5) node[pos=0.5] {$\delta$};
\draw[draw=black, fill=red, fill opacity=0.5, text opacity=1] (2.75,-1.5) rectangle ++(0.5,0.5) node[pos=0.5] {\footnotesize$\travelmin{1}{}$};
\draw[draw=black, fill=red, fill opacity=0.5, text opacity=1] (2.5,-1.5) rectangle ++(0.25,0.5) node[pos=0.5] {$\delta$};
\draw[draw=black, fill=red, fill opacity=0.5, text opacity=1] (2,-1.5) rectangle ++(0.5,0.5) node[pos=0.5] {\footnotesize$\travelmin{0}{}$};

\end{tikzpicture}
}
\caption{Illustration of the sequence of events that leads to the worst case upper and lower bounds in Theorem~\ref{thm:wcht=wchtt} for a route with 3 stops and 3 vehicles.}
\label{fig:thmwcht=wcht}
\end{figure}

\subsection{Headway Bounds for Routes with Any Set of Policies}
Further analysis is required to make a meaningful statement about headway bounds in routes with stops under $\dynamicpolicy$.
By our definition $\nullpolicy\subset\dynamicpolicy$ and strict monotonicity with respect to headway is not required, so there is no guarantee that bunch times increase on the output of a stop governed by $\dynamicpolicy$.
Instead, we first focus on identifying the condition that leads to maximum headway (when vehicles are as tightly bunched as possible) and then show the conditions that must occur for bunch time to be minimized.

\begin{lemma}
\label{lem:earlyarrival}
Given that the $j-1$st departure from $\busstop{i-1}$ occurs at fixed time $\departure{i-1}{j-1}$, a later arrival of the $j$th vehicle at $\busstop{i-1}$ can always lead to a later arrival of the $j$th vehicle at $\busstop{i}$, i.e.:
\begin{equation}
\label{eq:lemearlyarrival}
\arrival{i-1}{j}\leq\arrivalprime{i-1}{j}\implies\arrivalmax{i}{j}\leq\arrivalmaxprime{i}{j}
\end{equation}
\end{lemma}
\begin{proof}
When $\dwellmax{i-1}{}\geq\max\{\holdingprime{i-1}{j},\holding{i-1}{j}\}$, the maximum arrival time at $\busstop{i}$ is given by $\arrivalmax{i}{j}=\arrival{i-1}{j}+\dwellmax{i-1}{}+\travelmax{i-1}{}$ or $\arrivalmaxprime{i}{j}=\arrivalprime{i-1}{j}+\dwellmax{i-1}{}+\travelmax{i-1}{}$.
Algebraic manipulation reveals that $\arrivalmax{i}{j}\leq\arrivalmaxprime{i}{j}$.
When $\dwellmax{i-1}{}<\max\{\holdingprime{i-1}{j},\holding{i-1}{j}\}$, we need to show that the departure of the later arrival is always later, i.e., $\departure{i-1}{j}\leq\departureprime{i-1}{j}$.
While this is not true for any arbitrary policy, we show that it is true for the three policy classes considered in our formulation:
\begin{itemize}
    \item $\policy{i-1}\in\nullpolicy$: $\holding{i-1}{j}=\holdingprime{i-1}{j}$. This means departure time equals arrival time for both $\arrival{i-1}{j}$ and $\arrivalprime{i-1}{j}$, thus $\departure{i-1}{j}\leq\departureprime{i-1}{j}$.
    \item $\policy{i-1}\in\schedpolicy$: Consider the next scheduled release time $\mathbf{T}_l$.  For the $j$th arrival: $\holding{i-1}{j}=\max\{0,\mathbf{T}_l-\arrival{i-1}{j}\}$.  For $\arrivalprime{i-1}{j}\geq\arrival{i-1}{j}$, $j$ will either have the same scheduled release time or a later one: $\holdingprime{i-1}{j}=\max\{0,\mathbf{T}_{l+n}-\arrivalprime{i-1}{j}\}|n\in0\dots\infty$.
    \item $\policy{i-1}\in\dynamicpolicy$: Recall that the dynamic policy is 1-left-Lipschitz with respect to headway.  Since the lead vehicle $j-1$ departs at $\departure{i-1}{j-1}$ regardless if $j$ arrives at $\arrival{i-1}{j}$ or $\arrivalprime{i-1}{j}$, headway is only impacted by arrival time.  Ignoring dwell time, substituting $\arrival{i-1}{j} - \departure{i-1}{j-1}$ and $\arrivalprime{i-1}{j}-\departure{i-1}{j}$ for headway into the definition of a 1-left-Lipschitz function gives the result $\arrival{i-1}{j}+\holding{i-1}{j}\leq \arrivalprime{i-1}{j} + \holdingprime{i-1}{j}$.  Thus even if the dynamic policy holds a vehicle arriving later longer ($\holdingprime{i-1}{j}>\holding{i-1}{j}$), at the very latest, the vehicle will still depart $\busstop{i-1}$ at the same time as if it had arrived earlier, i.e., $\departure{i-1}{j}\leq\departureprime{i-1}{j}$.
\end{itemize}
Even if a vehicle departing at $\departure{i-1}{j}$ takes $\travelmax{i-1}{j}$, a vehicle departing at $\departureprime{i-1}{j}$ could do the same, meaning $\arrivalmax{i}{j}\leq\arrivalmaxprime{i}{j}$.
\end{proof}

\begin{theorem}
\label{thm:mingap}
The minimum time for a bunch of $M$ vehicles to pass a stop $\busstop{i}$ is a necessary condition for the maximum headway between the $j+M-1$st departure and the $j+M$th arrival for any $j\in0\dots\infty$ at that same stop.
\end{theorem}
\begin{proof}
The headway between the last vehicle in the bunch and the next arrival is given by $\headway{i}{j+M}=\arrival{i}{j+M}-\departure{i}{j+M-1}$.
If we fix the $j+M-1$st vehicle's departure at $\departure{i}{j+M-1}$, then by Lemma~\ref{lem:earlyarrival} the latest  possible arrival $\arrivalmax{i-1}{j+M}$ is a necessary condition for $\arrivalmax{i}{j+M}$, which maximizes $\headway{i}{j+M}$.
Likewise, $\arrivalmax{i-n}{j+M}$ is a necessary condition for $\arrivalmax{i-n+1}{j+M}$ all the way back to stop $\busstop{0}$.
Assuming all $M$ vehicles were released from $\busstop{0}$ (which is a requirement of our system model), then $\arrivalmax{0}{j+M}$ corresponds to $\arrivalmax{N}{j}$.
Again, by Lemma~\ref{lem:earlyarrival}, $\arrivalmax{N-n}{j}$ is a necessary condition for $\arrivalmax{N-n-1}{j}$ all the way back to stop $\busstop{i}$.
Thus, $\arrivalmax{i}{j}$ is a necessary condition for the maximum headway to occur.
Under Definition~\ref{def:bunch}, the time taken for a bunch of $M$ vehicles to pass the stop is given by $\departure{i}{j+M-1}-\arrival{i}{j}$.
This value is minimized when $\arrival{i}{j}=\arrivalmax{j}{j}$.
\end{proof}
From Theorem~\ref{thm:mingap}, it is clear that in addition to the condition that the time for a bunch of $M$ vehicles to pass a stop is minimized, another condition exists: the lead vehicle in the bunch must always achieve the maximum arrival time ($\arrivalmax{n}{j}|i<n\leq N$ and $\arrivalmax{n}{j+M}|0\leq n\leq i$) for all stops in the route until it comes back to service $\busstop{i}$.
Together these two conditions are sufficient to guarantee $\wcht{i}$.
We now must show what conditions are required to cause the minimum bunch time to occur.

\begin{lemma}
\label{lem:nxtmingap}
Given a bunch of $v$ vehicles with leader arriving at $\busstop{i}$ at time $\arrival{i}{j}$ and last vehicle departing $\busstop{i}$ at $\departure{i}{j+v-1}$, a smaller bunch time at $\busstop{i}$ always leads to a smaller possible bunch time at $\busstop{i+1}$, i.e.:
\begin{multline}
\departure{i}{j+v-1}-\arrival{i}{j}\leq\departureprime{i}{j+v-1}-\arrival{i}{j}\implies\\\departure{i+1}{j+v-1}-\arrival{i+1}{j-1}\leq\departureprime{i+1}{j+v-1}-\arrival{i+1}{j-1}
\end{multline}
\end{lemma}
\begin{proof}
We fix $\arrival{i}{j}$ and compare the bunch time at $\busstop{i+1}$ for two bunches with different final vehicle departure times $\departure{i}{j+v-1}\leq\departureprime{i}{j+v-1}$.
To minimize bunch time at $\busstop{i+1}$ we need to maximize the leader's arrival ($\arrival{i+1}{j}=\arrivalmax{i+1}{j}$) and minimize the last vehicle's departure ($\departure{i+1}{j+v-1}=\departuremin{i+1}{j+v-1}$ and $\departureprime{i+1}{j+v-1}=\departureminprime{i+1}{j+v-1}$).
$\arrival{i+1}{j}$ is maximized when dwell time at $\busstop{i}$ and travel to $\busstop{i+1}$ are maximized: $\dwell{i}{j}=\dwellmax{i}{}$ and $\travel{i}{j}=\travelmax{i}{j}$.
Note that the policy at $\busstop{i}$ may hold the vehicle for longer, but the effect is the same since we assumed the same lead vehicle arrival for both cases and holding is not affected by future arrivals.
To minimize $\departure{i+1}{j+v-1}$ and $\departureprime{i+1}{j+v-1}$ we need vehicles $j+1\dots j+v-1$ to travel to $\busstop{i+1}$ as fast as possible without overtaking the lead vehicle.
This is $\arrival{i+1}{j}-\departure{i}{j+k}$ and $\arrival{i+1}{j}-\departureprime{i}{j+k}$ for $k=1\dots v-1$, which is bounded from above and below by $\travelmax{i}{}$ and $\travelmin{i}{}$.
At this point we can see that $\arrival{i+1}{j+v-1}\leq\arrivalprime{i+1}{j+v-1}$ regardless of the positions of the other vehicles in the bunch.
Likewise, while waiting at $\busstop{i+1}$, the gap can be narrowed if $\dwell{i+1}{j+k}=\departure{i+1}{j}-\arrival{i+1}{j+k}$ and $\dwellprime{i+1}{j+k}=\departure{i+1}{j}-\arrivalprime{i+1}{j+k}$, which are of course bounded above by $\dwellmax{i+1}{}$ and below by $\dwellmin{i+1}{}$.
At this point, we see that the only way for $\departure{i+1}{j+v-1}\geq\departureprime{i+1}{j+v-1}$ is if $\holding{i+1}{j+v-1}\geq\holdingprime{i+1}{j+v-1}+\departureprime{i}{j+v-1}-\departure{i}{j+v-1}$.
Considering the same three cases as in Lemma~\ref{lem:earlyarrival}, however, shows us that this is impossible: under all three policy classes, a later arrival at the stop necessitates a later departure.
\end{proof}

In Lemma~\ref{lem:nxtmingap} we not only showed that a smaller bunch time at a preceding stop is a necessary condition for a smaller bunch time at a following stop, we showed the travel times and dwell times required to achieve it for each vehicle in the bunch.
We can now describe how to calculate the smallest bunch time at any given stop for a bunch of $M$ vehicles.

\begin{theorem}
\label{thm:wcht}
For a bunch of $M$ vehicles, bunch time is minimized at $\busstop{i+1}$ when, starting at $\busstop{0}$, the lead vehicle in the bunch $j$ dwells with $\dwell{i}{j}=\dwellmax{i}{}$ and travels with $\travel{i}{j}=\travelmax{i}{}$ for all stops and the subsequent $M-1$ vehicles always travel with parameters $\travel{i}{j+m}=[\arrival{i+1}{j}-\departure{i}{j+m}]^{\travelmax{i}{}}_{\travelmin{i}{}}$ and $\dwell{i+1}{j+m}=[\departure{i+1}{j}-\arrival{i+1}{j+m]}]^{\dwellmax{i+1}{}}_{\dwellmin{i+1}{}} \forall m\in1\dots M; j\in0,M,2M,\dots\infty; i\in0\dots N$.
\end{theorem}
\begin{proof}
We start with the base case where vehicles are released from $\busstop{0}$ at $\release{0}{0}\dots\release{0}{M-1}$.
Since there is only one possible departure time for each vehicle, this case only has one possible bunch time, and is therefore minimized.
We now show that for a bunch of $M$ vehicles with the first vehicle arriving at $\busstop{i}$ at $\arrival{i}{j}$ and the last vehicle departing at $\departure{i}{j+M-1}$, minimum bunch size at $\busstop{i+1}$ can only occur if $\departure{i}{j+M-1}-\arrival{i}{j}$ is already minimized.
By Lemma~\ref{lem:nxtmingap}, we know this is true for the same bunch of $M$ vehicles, however, we still need to show that this is true for all future bunches of $M$ vehicles.
We prove this by contradiction: assume that for a given $\arrival{i}{j}$, there is some $\departure{i+1}{j+M-1+x}-\arrival{i+1}{j+x}<\departure{i+1}{j+M-1}-\arrival{i+1}{j}$ for $x>0$.
If this was the case, then by Lemma~\ref{lem:nxtmingap} it must also hold true that $\departure{i}{j+M-1+x}-\arrival{i}{j+x}<\departure{i}{j+M-1}-\arrival{i}{j}$, however, this would violate the assumption that $\departure{i}{j+M-1}-\arrival{i}{j}$ is minimized: a contradiction.
Since the bunch size at the start is minimized in the base case, we can use the travel and dwell parameters computed in Lemma~\ref{lem:nxtmingap} to recursively calculate the minimum bunch size at $\busstop{1}\dots\busstop{N}$.
After $\busstop{N}$, the vehicles travel back to $\busstop{0}$ and begin traveling the route again.
Now that all the stops have been visited by a vehicle at least once, headway is computable.
Minimum bunch time is now calculated for the bunch where the lead vehicle is the $M$th arrival and the last vehicle is the $2M-1$st arrival at each stop.
By our induction hypothesis, this continues for all bunches of $M$ vehicles for all $N$ stops with the lead vehicle arriving at $M,2M,\dots\infty$.
\end{proof}

Using Theorem~\ref{thm:wcht}, we can now compute the  headway upper bound.
We note that the travel and dwell times taken by the vehicles in Theorem~\ref{thm:wcht} are sufficient to cause the maximum headway to occur, but not always necessary: the same result can sometimes be reached with different travel and dwell parameters.
While it is not true in general that a smaller bunch time of $M$ vehicles implies a lower minimum headway time across all $M$ vehicles in a bunch (e.g., all vehicles in the bunch are evenly spaced instead of two successive vehicles having zero headway), we now show that by following the travel and dwell times prescribed by Theorem~\ref{thm:wcht}, the lower bound on headway time must occur.

\begin{theorem}
\label{thm:bcht}
The headway lower bound, $\bcht{i}$ occurs at $\busstop{i}$, when, starting at $\busstop{0}$, the lead vehicle $j$ in a bunch of $M$ vehicles dwells with $\dwell{i}{j}=\dwellmax{i}{}$ and travels with $\travel{i}{j}=\travelmax{i}{}$ for all stops and the subsequent $M-1$ vehicles always travel with parameters $\travel{i}{j+m}=[\arrival{i+1}{j}-\departure{i}{j+m}]^{\travelmax{i}{}}_{\travelmin{i}{}}$ and $\dwell{i+1}{j+m}=[\departure{i+1}{j}-\arrival{i+1}{j+m]}]^{\dwellmax{i+1}{}}_{\dwellmin{i+1}{}} \forall m\in1\dots M; j\in0,M,2M,\dots\infty; i\in0\dots N$.
\end{theorem}
\begin{proof}
To minimize $\bcht{i}$, we want $\departure{i}{j-1}=\departuremax{i}{j-1}$ and $\arrival{i}{j}=\arrivalmax{i}{j}$.
The base case is the departure of the $0$th vehicle and the arrival of the $1$st vehicle at $\busstop{1}$. Under the travel conditions, listed above, it is clear that $\departure{1}{0}$ is maximized and $\arrival{1}{1}$ is minimized since $\travel{0}{0}=\travelmax{0}{}$ and $\travel{0}{1}=[\arrival{1}{0}-\departure{0}{1}]^{\travelmax{0}{}}_{\travelmin{0}{}}$ will clamp the difference between these values as close to $0$ as possible.
The induction hypothesis is that the smallest possible headway for the $j$th arrival at $\busstop{i}$ leads to the smallest possible headway for the $j$th vehicle at $\busstop{i+1}$ under the given travel conditions.
It is clear that $\arrival{i+1}{j-1}=\departure{i}{j-1}+\travelmax{i}{}$ is the latest possible arrival given the starting condition and any other travel time cannot lead to a later departure $\departure{i+1}{j-1}$ under Lemma~\ref{lem:earlyarrival}.
Meanwhile, the $j$th arrival at $\busstop{i}$ dwells for $\dwell{i}{j}=[\departure{i}{j-1}-\arrival{i}{j]}]^{\dwellmax{i}{}}_{\dwellmin{i}{}}$, which results in its $\departure{i}{j}$ being as close to $\departure{i}{j-1}$ as possible without overtaking.  The travel condition $\travel{i}{j}=[\arrival{i+1}{j-1}-\departure{i}{j}]^{\travelmax{i}{}}_{\travelmin{i}{}}$ then minimizes $\arrival{i+1}{j}$ given that $\arrival{i+1}{j-1}$ is fixed (if an overtake were to occur, the $j$th departure would become the $j-1$st arrival, so traveling at $\travelmin{i}{}$, would not improve $\arrival{i+1}{j}$).
Since $\departure{i+1}{j-1}$ is maximized and $\arrival{i+1}{j}$ is minimized, $\headway{i+1}{j}$ will be minimized given the initial condition.
From the base case, this is minimized for $\headway{1}{1}$.
\end{proof}

In Theorem~\ref{thm:bcht}, we see that the lower bound on headway can be calculated using a similar process to the upper bound but instead focusing on arrival of the $1,M+1,2M+1,\dots$ vehicle at every stop.
This means $\wcht{i}$ and $\bcht{i}$ can be calculated simultaneously.

\subsection{Computing Headway Bounds}
Algorithm~\ref{alg:wcht} uses Theorem~\ref{thm:wcht} and Theorem~\ref{thm:bcht} to calculate the headway upper and lower bounds.
It requires the stop parameters ($\travelmin{i}{}, \travelmax{i}{}, \dwellmin{i}{}, \dwellmax{i}{}, \policy{i}$) for all $N$ stops in a route.
To satisfy the base case in Theorem~\ref{thm:wcht} and Theorem~\ref{thm:bcht}, the initial release times $\release{0}{j}$ of all $M$ vehicles are required as well.
$i$ is used as an index to track the current stop for which headway time is being calculated (initialized in Line 1 and updated in Line 14).
Initially, $\wcht{i}$ is set to $-\infty$ for all stops, since it has not yet been computed (Line 2).
A negative value indicates no valid headway while $\wcht{i}\geq0$ indicates the worst possible headway time at $\busstop{i}$ up to the time of the latest considered departure.
Likewise, $\bcht{i}$ is initially set to $\infty$ for all stops (Line 3) and decreases as the loop progresses.
We use $\carryin{i}$ to track the departure time of the last bus at each stop $\busstop{i}$.
This is initialized to undefined for each stop since at the start of the route, no buses have departed (Line 4).
Two vectors $\mathbf{a}$ and $\mathbf{d}$ track the latest arrivals and departures for all $M$ buses at the stop $i$ under consideration.
We know the initial departures for all buses (corresponding to the base case in Theorem~\ref{thm:wcht}), so $\mathbf{d}$ is initialized to release times $\release{0}{0}\dots\release{0}{M}$ (Line 5).
Following Theorem~\ref{thm:wcht}, we loop through the successive stops calculating arrivals and departures given that buses travel with $\travel{i}{}$ and $\dwell{i}{}$ determined by Lemma~\ref{lem:nxtmingap}.
We use the variable $\mathtt{converged}$ (initially set to false in Line~6) to track the convergence on headway bounds.
Inside the loop, the first element of $\mathbf{a}$ (corresponding to the first arrival at $\busstop{i}$) receives the earliest previously tracked departure plus the maximum travel time (Line 9).
For the remaining vehicles the travel time is computed to be as fast as possible, without over taking the lead (Lines 10--12).
Next, we clear the previous departure values (Line 13) and update the stop index (Line~14).
The lead vehicle's dwell time is set as long as possible (Line~15), and the remaining vehicles are set to dwell such that the times between their departures and the lead vehicle's are minimized (Line 17).
Note that calculation of $\holding{i}{}$ may require additional variables not tracked here, for example, a stop controlled by $\schedpolicy$ will require the current time since start of service and a stop controlled by $\dynamicpolicy$ will require an estimate of the next vehicle's arrival time $\estimatedarrival{i}{j+1}$.
This estimate needs to be for the earliest possible arrival of the subsequent vehicle so that the bunch time will shrink and the algorithm captures the worst-case behavior.
Once departure times have been calculated, we can update our headway bounds: if a stop has been previously visited ($\carryin{i}$ exists -- Line 19) and the $\wcht{i}$ is larger than previously observed (Line 20), we record it (Line 21).
Convergence of $\wcht{i}$ depends only on bunch time: the same bunch time of $M$ vehicles always leads to the same next minimum bunch time, so if there is no change, it is possible that our algorithm has converged (Line 23).
We also update $\bcht{i}$ if this is lower than previously observed (Line 24).
It is possible that $\bcht{i}$ converges slower than $\wcht{i}$ since the positions of vehicles in a bunch can continue changing even if bunch time remains unchanged, thus an update to $\bcht{i}$ causes us to unset the convergence flag (Line~26).
Before considering the next stop, we set $\carryin{i}$ to the departure of the last vehicle, which will be used to calculate $\wcht{i}$ when the bunch visits $\busstop{i}$ again (Line 27).

\begin{algorithm}
\caption{Dynamic program to calculate $\wcht{i}$ and $\bcht{i}$.}\label{alg:wcht}
\KwData{Stop params $\travelmin{i}{},\travelmax{i}{},\dwellmin{i}{},\dwellmax{i}{},\policy{i}\forall i\in0\dots N$\\Vehicle release times $\release{0}{j} \forall j\in 0\dots M$}
\KwResult{$\wcht{i},\bcht{i}\forall i\in0\dots N$}
$i\gets1$;\\
$\wcht{i}\gets-\infty\forall i\in0\dots N$;\\
$\bcht{i}\gets\infty\forall i \in0\dots N$;\\
$\carryin{i}\gets \mathtt{undef.}\forall i\in0\dots N$;\\
$d_j\gets \release{0}{j}\forall j\in0\dots M$;\\
$\mathtt{converged}\gets \bot$;\\
\While{$\lnot\mathtt{converged}$}{
    $\mathbf{a}\gets[]$;\\
    Push $d_0 + \travelmax{i}{}$ to $\mathbf{a}$;\\
    \For{$j\in 1\dots M$}{
        Push $d_j + [a_{j-1}-d_j]^{\travelmax{i}{}}_{\travelmin{i}{}}$ to $\mathbf{a}$;\\
    }
    $\mathbf{d}\gets[]$;\\
    $i \gets (i + 1)\mod N$;\\
    Push $a_0 + \max\{\dwellmax{i}{},\holding{i}{}\}$ to $\mathbf{d}$;\\
    \For{$j\in 1\dots M$}{
        Push $a_j + \max\{[d_{j-1}-a_j]^{\dwellmax{i}{}}_{\dwellmin{i}{}}, \holding{i}{}\}$ to $\mathbf{d}$;\\
    }
    \uIf{$\exists \carryin{i}$}{
        \uIf{$a_0-\carryin{i}>\wcht{i}$}{
            $\wcht{i}\gets \max\{0,a_0-\carryin{i}\}$;\\
        }
        \uElse{
            $\mathtt{converged}\gets\top$;\\
        }
    }
    \uIf{$\max\{0, a_1-d_0\}<\bcht{i}$}{
        $\bcht{i}\gets\max\{0,a_1-d_{0}\}$;\\
        $\mathtt{converged}\gets\bot$;\\
    }
    $\carryin{i}\gets d_M$;\\
}
\end{algorithm}

We note that Algorithm~\ref{alg:wcht} is not guaranteed to converge in a finite amount of time.
Consider, for example, a stop controlled by $\policy{i}\in\dynamicpolicy$, where holding times are calculated based on headway such that $\holding{i}{j}=\headway{i}{j}/t$ (where $t$ is time since service start) and $\dwellmin{i}{}=\dwellmax{i}{}=0$.
This policy still respects all the requirements of $\dynamicpolicy$, but the bunch time at $\busstop{i}$ can only approach $0$: it will never reach it in a finite time window.
In practice however, a time cutoff can be introduced since most routes do not serve stops continuously and pause at night or on the weekends.
Since convergence is not provable, we consider the time complexity to update the bounds for one stop, which is $O(NM)$ because arrivals and departures must be calculated for $M$ vehicles at each of the $N$ stops leading to the update.
Memory complexity is linear with respect to number of stops and vehicles ($O(NM)$). The vectors of arrivals and departures ($\mathbf{a}$ and $\mathbf{d}$) never exceed length $M$ and are freed and reallocated on each iteration of the loop.
The last departure for each stop, $\omega$, is also required for all $N$ stops.

\section{Experiments}
To build intuition on the effects of different holding policies on headway bounds, we first demonstrate our algorithm on an example bus route with $5$ stops.
All code for the following sections is publicly available on GitHub\footnote{https://github.com/smarttransit-ai/bus-transit-headway-bounds}.
The travel and dwell time parameters for each stop are $\travel{i}{}\in[4,4.5]$ and $\dwell{i}{}\in[0,0.5]$, with units given in minutes.
This yields a minimum route completion time of 20 min. and a maximum route completion time of 25 min. in the uncontrolled case, making it easier to link the computed headway times to other critical timing parameters.
We compare the following policies:
\begin{itemize}
    \item $\nullpolicy$ -- $\holding{i}{j}=0\forall i\in0\dots N,j\in 0\dots\infty$.
    \item $\schedpolicy$ -- We set scheduled departure times at intervals of the desired headway ($10$ min.) and add an offset for the maximum time it takes the lead bus to reach a stop after beginning the route.  For example, the schedule list at $\busstop{1}$ would be $\mathbf{T}=\langle5,15,25,\dots\rangle$, since the maximum time to reach $\busstop{1}$ is 5 minutes ($\travelmax{i}{}+\dwellmax{i}{}$). This offset is chosen to prevent the worst-case scenario for a schedule-driven policy outlined in Note~\ref{note:schedpolicy}.
    \item $\dynamicpolicy$ -- We use the same headway-based dynamic holding policy as Zhou et al.~\cite{zhou_a_2022}, adapting their notation to match our system model in Eq.~(\ref{eq:policy}).
\begin{equation}
\label{eq:policy}
\holding{i}{j}=
\begin{cases}
\Bigg[\frac{1}{2}(\estimatedarrival{i}{j+1} + \departure{i}{j-1} - 2\arrival{i}{j})\Bigg]^\Phi_0 & \text{if} \quad\frac{\headway{i}{j}}{\headway{i}{j+1}} \leq R\\
0 & \text{else}
\end{cases}
\end{equation}
    In addition to the arrivals and departures of the surrounding vehicles, the policy is parameterized by $\Phi$, the maximum amount of holding that can be applied, and $R$, the headway to tailway ratio needed to activate the policy.  These parameters are desirable for transit operators as they bound holding times from above and below.
\end{itemize}

\textbf{Holding policies at every stop} --
We begin by examining a route where all $5$ stops have an identical holding policy applied, i.e., all stops governed by either $\nullpolicy$, $\schedpolicy$, or $\dynamicpolicy$, in Fig.~\ref{fig:all_stops}. 
The route is served by $3$ buses released from $\busstop{0}$ at times $0$, $10$, and $20$ min.
Since the maximum completion time of the uncontrolled route is $30$ min., three buses should be sufficient to maintain a headway of 10 min. at each stop.
In Fig.~\ref{fig:all_stops}, the hatched bars represent the upper and lower headway time bounds and the box plots represent the headways actually observed in a probabilistic simulation of 10,000 trips along the route with travel and dwell times drawn uniformly at each stop.
As predicted in Theorem~\ref{thm:wcht=wchtt}, when there is no control policy, worst case bunching occurs and $\wcht{i}$ approaches the maximum route completion time while $\bcht{i}$ approaches zero.
Furthermore, we observe that the statistical simulation also approaches this bound.
When all stops are scheduled so that the driver has to hold at a stop if they arrive early, we see that headway times are maintained around 10 minutes with tight variance, except for $\busstop{0}$.
This is because the schedule release times at $\busstop{4}$ are $\langle20, 30, 40\dots\rangle$ while at $\busstop{0}$ they are at $\langle30, 40, 50,\dots\rangle$.
In the average case, buses will reach $\busstop{0}$ around $5$ minutes after their predecessor left.
We observe that despite the relatively tight statistical distribution of headway, maintaining a schedule can incur up to $20$ min. of headway ($\wcht{i}$) at all stops other than $\busstop{0}$ due to follower buses in a bunch becoming unacceptably delayed.
For $\dynamicpolicy$ we set $R=0.75$ and $\Phi=30$.
In this case, the mean headway time is less than that of the schedule driven route although the average variance at each stop is higher.
For all stops excluding $\busstop{0}$, the headway upper and lower bounds are less than $\schedpolicy$ meaning that the $\dynamicpolicy$ is able keep buses evenly spaced allowing headway to approach the average travel time divided by number of buses, which is lower than the $\schedpolicy$'s strictly enforced $10$ min.

\begin{figure}
    \centering
    \includegraphics[width=1\linewidth,trim={1.5cm 0.25cm 3cm 2cm},clip]{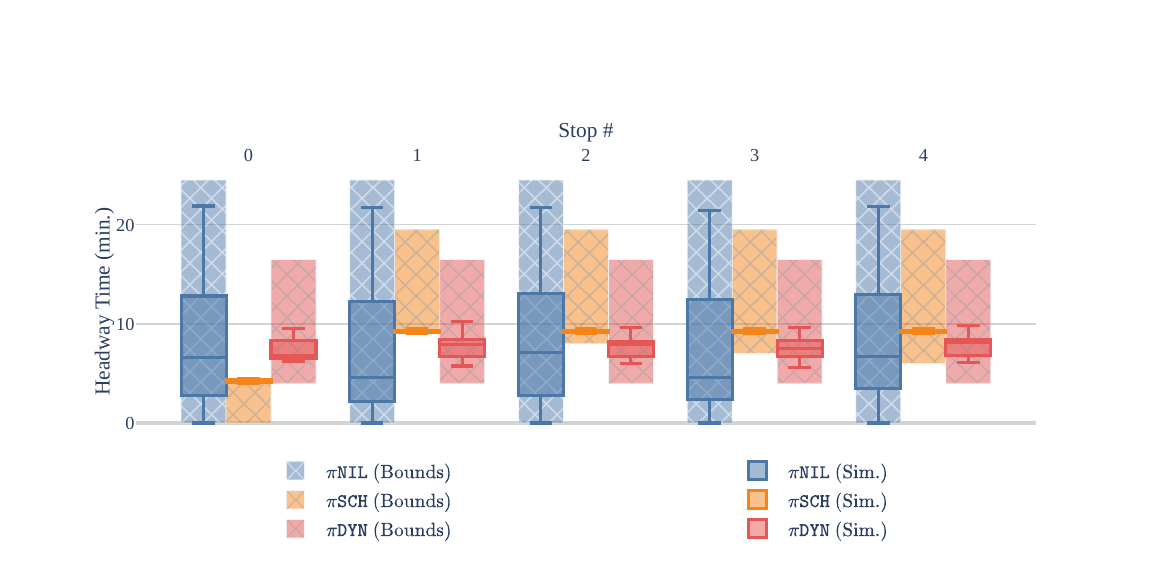}
    \caption{Simulated headway measurements (boxes) and headway bounds (hatched bars) for a $5$ stop route with all stops controled by $\nullpolicy$ (blue), $\schedpolicy$ (orange), or $\dynamicpolicy$ with $R=0.75$, $\Phi=30$ (red).}
    \label{fig:all_stops}
\end{figure}

\textbf{Holding policies at one stop} --
Some works argue that a single control point is sufficient to regulate headway~\cite{eberlein_the_2001}, which we analyze in Fig.~\ref{fig:first_stop} for holding only applied at $\busstop{0}$ using the same $\schedpolicy$ and $\dynamicpolicy$ parameters as Fig.~\ref{fig:all_stops}.
We see that for $\schedpolicy$, the bounds are equivalent to when all stops are controlled by $\schedpolicy$, but that the variance of the simulated results is now greater.
However, for the $\dynamicpolicy$ route, $\bcht{i}$ is lower for all stops, meaning $2$ of the $3$ vehicles are approaching each other.
Despite this, $\wcht{i}$ remains unchanged meaning one stop is sufficient to maintain separation.

\begin{figure}
    \centering
    \includegraphics[width=1\linewidth,trim={1.5cm 0.25cm 3cm 2cm},clip]{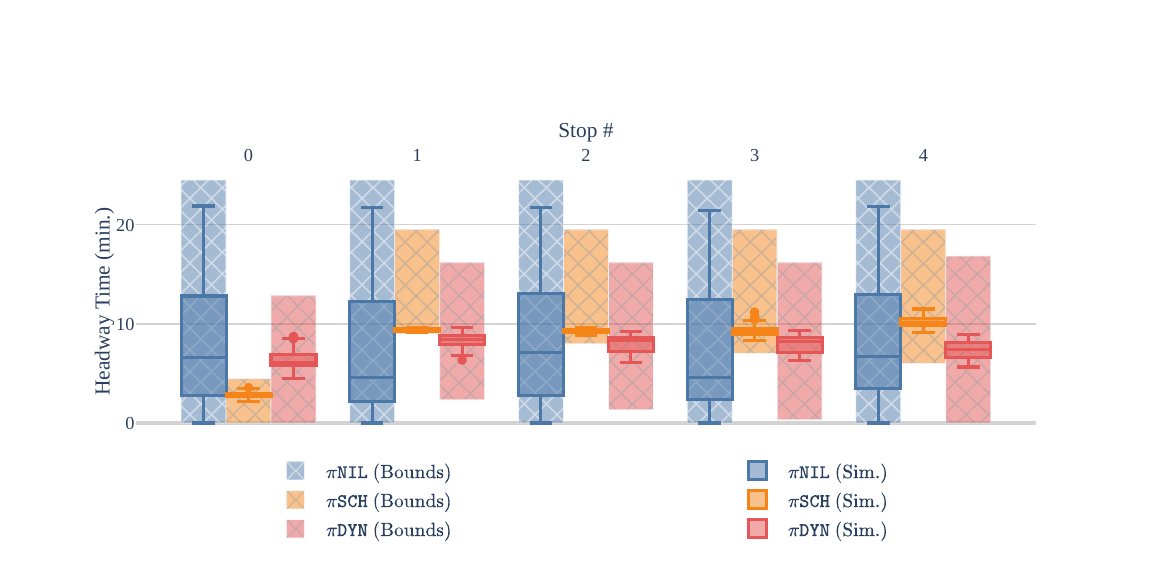}
    \caption{Simulated headway measurements (boxes) and headway bounds (hatched bars) for a $5$ stop route with all stops governed by $\nullpolicy$ (blue) except for $\busstop{0}$, which can be governed by $\schedpolicy$ (orange) or $\dynamicpolicy$ with $R=0.75$, $\Phi=30$ (red).}
    \label{fig:first_stop}
\end{figure}

\textbf{Combinations of holding policies} --
Other authors argue that multiple control points lead to better performance~\cite{chen_implementation_2013}.
We investigate this claim and also compare with the case when different strategies are combined on the same route.  In Fig.~\ref{fig:2stops}, $\busstop{0}$ and $\busstop{2}$ are control points and we consider (1) both stops controlled by $\schedpolicy$, (2) both stops controlled by $\dynamicpolicy$, (3) one stop controlled by $\schedpolicy$ ($\busstop{2}$) and one stop controlled by $\dynamicpolicy$ ($\busstop{0}$).
The policy parameters are the same as Fig.~\ref{fig:all_stops}.
Case (1) maintains control similar to its counterparts in Figs.~\ref{fig:all_stops} and~\ref{fig:first_stop} indicating that more scheduled stops are not necessary to increase performance in this case.
In case (2) we see that by adding a control point midroute, the $\bcht{i}$ improves for the stops, though not to the extent of Fig.~\ref{fig:all_stops}.
In case (3) we see that the bounds are tighter than (1) and (2) immediately following the $\schedpolicy$ control point showing that $\schedpolicy$'s ability to regulate headway is somewhat enhanced by the $\dynamicpolicy$ at $\busstop{0}$, despite the marginally worse headway on the output of that stop.
This shows that policies can be combined along a route to shape headway bounds, but that they may inadvertently negatively impact bounds at other locations as well.

\begin{figure}
    \centering
    \includegraphics[width=1\linewidth,trim={1.5cm 0.25cm 3cm 2cm},clip]{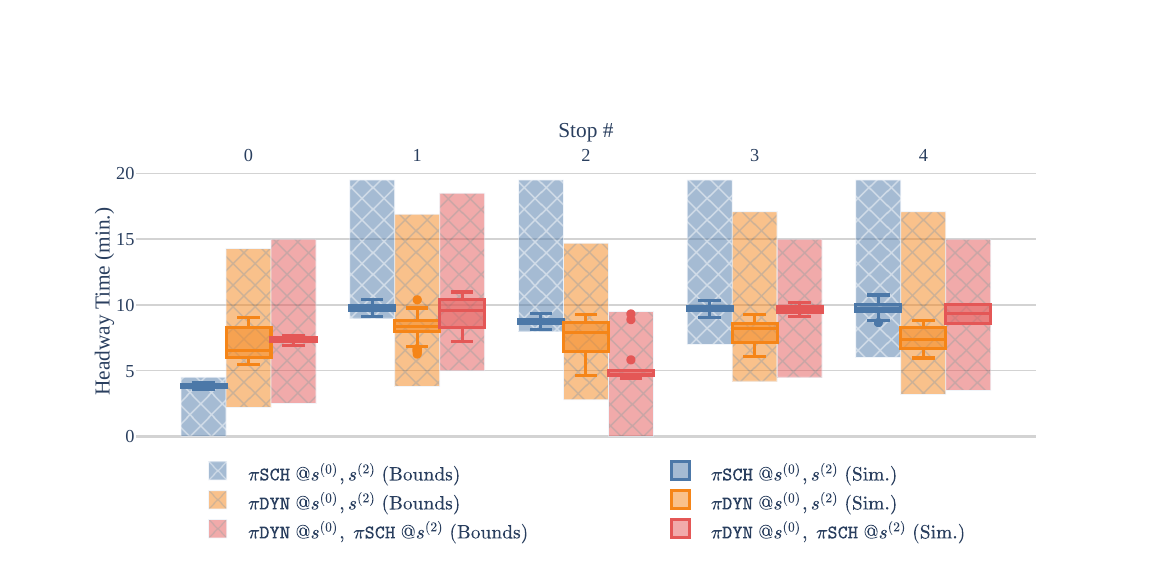}
    \caption{Simulated headway measurements (boxes) and headway bounds (hatched bars) for a $5$ stop route with holding polices applied at $\busstop{0}$ and $\busstop{2}$.  $\dynamicpolicy$ is applied with parameters $R=0.75$, $\Phi=30$.}
    \label{fig:2stops}
\end{figure}

\textbf{Hyperparameters of $\dynamicpolicy$} --
We now consider the route where $\dynamicpolicy$ is applied only at $\busstop{0}$ and adjust the maximum holding ($\Phi$) and the headway/tailway ratio at which the policy activates ($R$).
For transit operators, issuing holds of less than a minute can be difficult for drivers to execute (necessitating a lower $R$), and too much holding can delay passengers already onboard a vehicle (necessitating a bound on $\Phi$).
The results are shown in Fig.~\ref{fig:headway_hp}.
We observe that the best performance is when $R=1$, indicating that even minor deviations in headway will trigger holding, but that there is no difference in bounds with $R=0.75$, despite more variance in the simulated results.
Below $R=0.75$ the upper bound increases and the lower bound decreases indicating reduced performance, which is matched in the simulated results.
We also observe that decreasing $\Phi$ (reducing the amount of hold time the policy can issue), has little effect on the simulated performance, but causes the upper headway bound to increase.
As $\Phi\rightarrow0$, $\wcht{i}$ will approach the performance of $\nullpolicy$.

\begin{figure}
    \centering
    \includegraphics[width=1\linewidth,trim={1.5cm 0.25cm 3cm 2cm},clip]{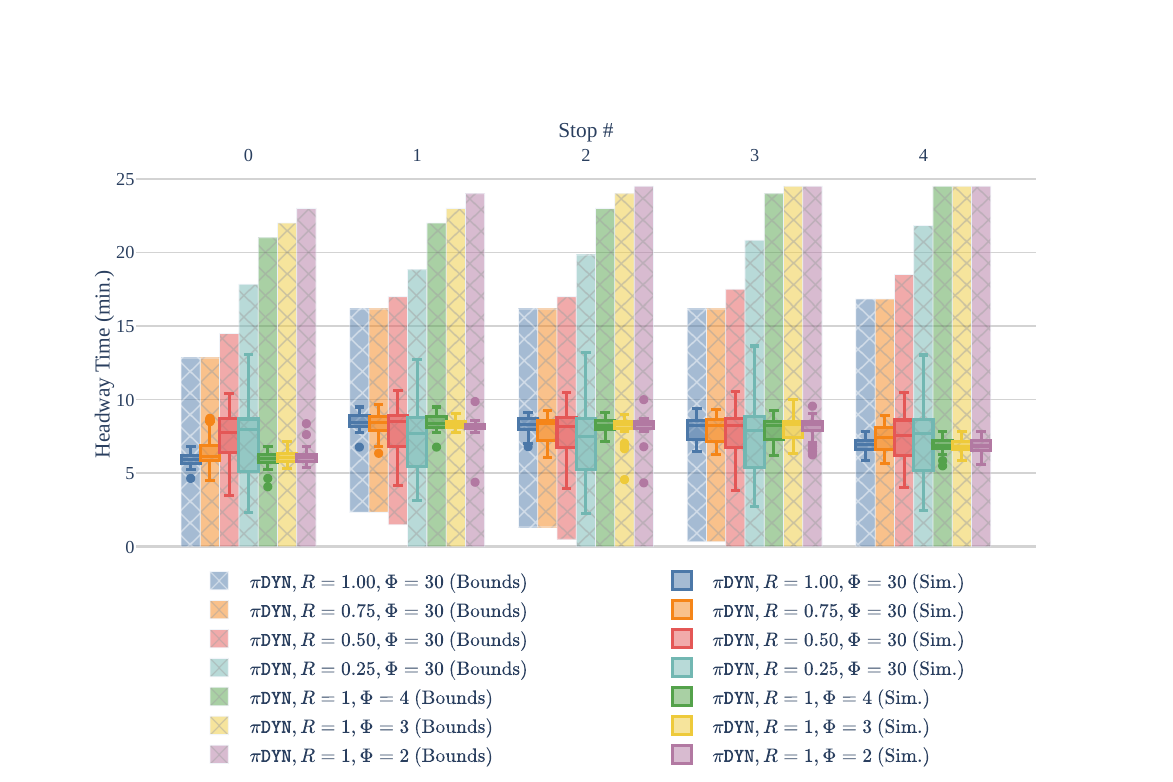}
    \caption{Simulated headway measurements (boxes) and headway bounds (hatched bars) for a $5$ stop route with $\dynamicpolicy$ applied at $\busstop{0}$ while sweeping parameters $R$ and $\Phi$.}
    \label{fig:headway_hp}
\end{figure}

\section{Case Study}
We now demonstrate the utility of our analytical bound when selecting holding policies for a real bus route in Nashville, TN.
Route 55 runs from Murfreesboro to downtown Nashville and consists of about 80 stops round-trip with 16 buses in service at a time, which leads to a high probability of bunching during peak commute times.
While we cannot arbitrarily test different holding policies for drivers and passengers, we have real world data for dwell times and travel times across all hours of the day.
These are fit to log-normal distributions as is common practice in transit system modeling \cite{rajbhandari_estimation_2003} of which we take the $5$th to $95$th percentile to generate our $\travelmin{i}{}$, $\travelmax{i}{}$, $\dwellmin{i}{}$, and $\dwellmax{i}{}$ for each stop.
In the existing route, all buses depart the Hickory Hollow terminal ($\busstop{0}$) and City Center terminal ($\busstop{40}$) according to a static schedule.
This is our baseline case in Fig.~\ref{fig:route55}, which shows the headway bounds, and case (A) in Fig.~\ref{fig:time_diff_comparison}, which shows the cumulative distributions of departure-to-departure times (a proxy for headway times) for all stops in a transit simulator internal to WeGo that comprehensively models passenger arrivals at stops, bus crowding, and vehicle overtakes.
We observe that in this case, buses can arrive behind schedule and the situation described in Note~\ref{note:schedpolicy} can occur leading to a headway upper bound of $120$ min.
The simulated results bear this out with a relatively flat distribution of waiting times with some as long as an hour.
In (B), we also include headway times for the ideal case, where we could apply $\dynamicpolicy$ at every stop with $R=1$ and no limit on $\Phi$.

Our goal is to replace the schedule with dynamic holding at key points along the route to regulate headway.
Due to physical constraints, this can only occur at Bell, Glengarry, and Thompson stations, whose positions along the route are marked in Fig.~\ref{fig:route55}.
In addition, we want to constrain total possible waiting time to $30$ min. to mitigate the probability of delayed releases from the terminals.
Thus, we consider two cases: (C) applying $\dynamicpolicy$ at Thompson station (inbound and outbound) with $\Phi=15$ min. in each direction, and (D) applying $\dynamicpolicy$ at all four potential control points but with $\Phi=7.5$ min.
(C) is operationally easier to implement, but we see that even with the limitation on $\Phi$, by spreading the holding over more control points (D) is able to outperform (C) in headway bounds along the middle of the route, which corresponds to the downtown core.
The simulated results bear a similar story with (D) marginally outperforming (C) in the long tail.
While these results show that our analysis can quickly identify headway bounds, it is important to remember that the specific result for this route -- that spreading holding over many control points outperforms two major control points -- may be different for different routes with different numbers of buses and scheduled departure times.
Our analysis technique is meant to facilitate exploration of such cases.

\begin{figure}
    \centering
    \includegraphics[width=1\linewidth,trim={1.25cm 0.75cm 2.25cm 2.25cm},clip]{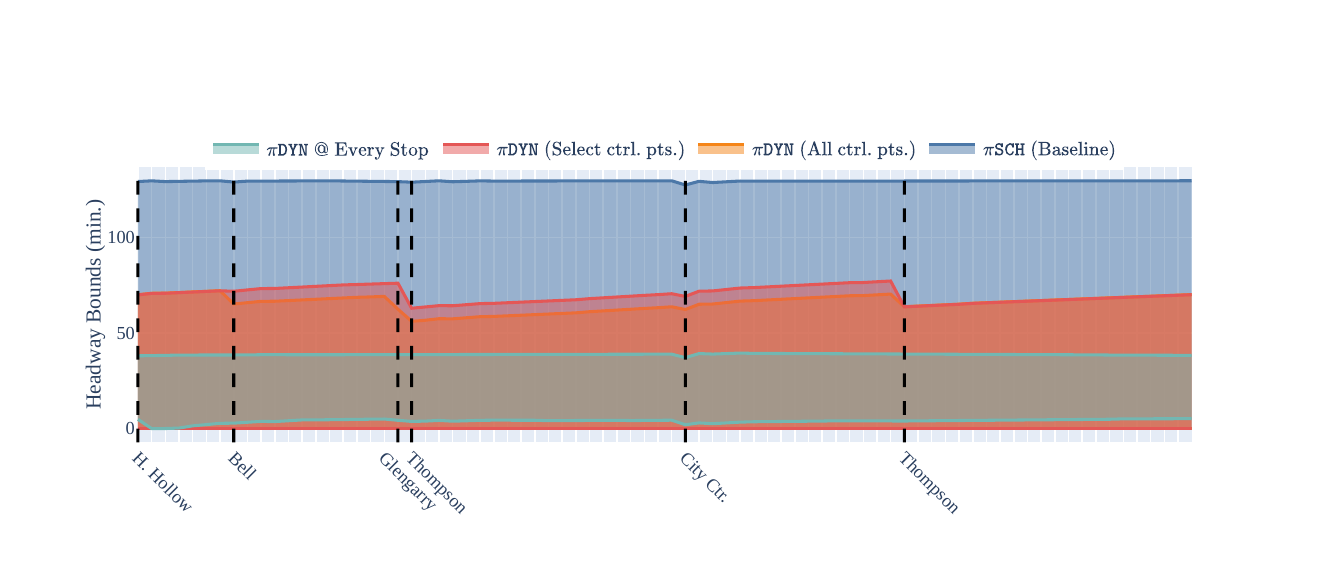}
    \caption{Headway bounds for WeGo route 55 in Nashville, TN under a schedule driven baseline (A), $\dynamicpolicy$ with $R=1$ and unbounded $\Phi$ applied at every stop (B), holding at 2 control points with $\Phi=15$ min. (C), and holding at 4 control points with $\Phi=7.5$ min. (D).}
    \label{fig:route55}
\end{figure}

\begin{figure}
    \centering
    \includegraphics[width=0.74\linewidth,clip,clip]{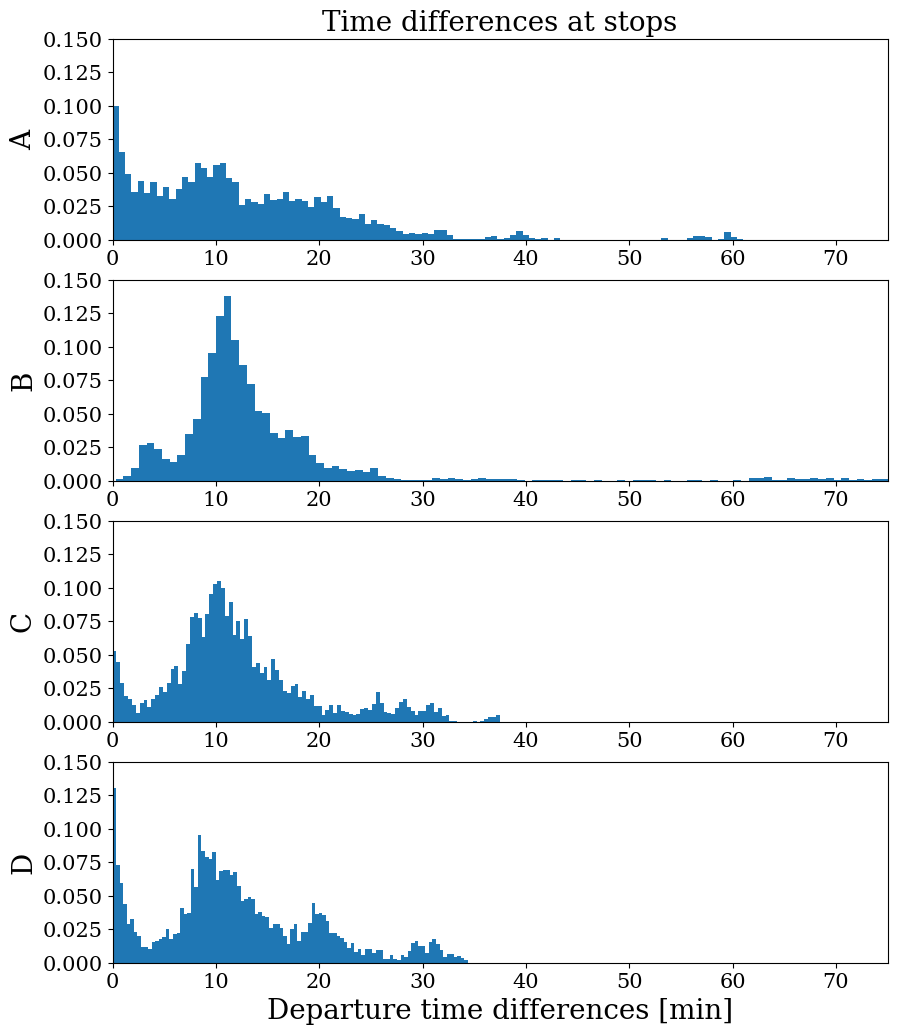}
    \caption{Simulated departure-to-departure times aggregated across all stops for the schedule-driven baseline (A), $\dynamicpolicy$ applied at every stop (B), holding at 2 control points with $\Phi=15$ min. (C), and holding at 4 control points with $\Phi=7.5$ min. (D).}
    \label{fig:time_diff_comparison}
\end{figure}

\section{Conclusion}
We have shown how to calculate upper and lower bounds on headway times under the assumption that travel and dwell times between stops in a transit route are bounded.
Through experiments on an example bus route, we have built some intuition about the affects of various policies and holding point locations on the route performance.
Finally, we have demonstrated the utility of this analysis by applying it to a real transit system and showing how the results can be applied in transit planning.
In the future, we hope to extend this work to include an upper bound on travel time on routes with arbitrary holding policies.
We would also like to develop a formal design strategy based on the results of this analysis that efficiently identifies optimal holding policies and stop locations for arbitrary routes.
We hope this work serves as a starting point for future worst-case analysis in transit planning.

\section*{Acknowledgment}
This work was supported in part by the U.S. Department of Transportation through the SMART Grants Program Phase~1 award in partnership with WeGo Public Transit and the Metropolitan Government of Nashville and Davidson County.  Any opinions, findings, conclusions, or recommendations expressed in this material are those of the authors and do not necessarily reflect the views of the U.S. Department of Transportation.

\bibliographystyle{IEEETran}
\bibliography{references}

\end{document}